\documentclass[reqno]{amsart}
\usepackage{amssymb}
\usepackage{amsfonts}

\newtheorem{theorem}{Theorem}
\theoremstyle{plain}

\newtheorem{claim}{Claim}

\newtheorem{definition}{Definition}

\newtheorem{lemma}{Lemma}

\newtheorem{proposition}{Proposition}
\newtheorem{remark}{Remark}

\numberwithin{equation}{section}
\input{tcilatex}

\begin{document}
\title[3-body second order correction and uniform error estimate]{Second
Order Corrections to Mean Field Evolution for Weakly Interacting Bosons in
The Case of 3-body Interactions}
\author{Xuwen Chen}
\address{Department of Mathematics\\
University of Maryland\\
College Park, MD 20742}
\email{chenxuwen@math.umd.edu}
\date{06/24/2011 (v3 for ARMA)}

\begin{abstract}
In this paper, we consider the Hamiltonian evolution of $N$ weakly
interacting Bosons. Assuming triple collisions, its mean field approximation
is given by a quintic Hartree equation. We construct a second order
correction to the mean field approximation using a kernel $k(t,x,y)\ $and
derive an evolution equation for $k$. We show the global existence for the
resulting evolution equation for the correction and establish an apriori
estimate comparing the approximation to the exact Hamiltonian evolution. Our
error estimate is global and uniform in time. Comparing with the work of
Rodnianski and Schlein\cite{RodnianskiAndSchlein}, and Grillakis, Machedon
and Margetis \cite{GMM1,GMM2}, where the error estimate grows in time, our
approximation tracks the exact dynamics for all time with an error of the
order $O\left( 1/\sqrt{N}\right) .$
\end{abstract}

\maketitle

\section{Introduction}

\label{section:intro}

In Bose-Einstein condensation (BEC), particles of integer spins
(\textquotedblleft Bosons\textquotedblright ) occupy a macroscopic quantum
state often called the \textquotedblleft condensate\textquotedblright . The
initial observation of Einstein and Bose was confirmed experimentally in
1995 and repeated later \cite{Anderson, Davis95, Stamper}. This phenomenon
has stimulated the study of the theory of many-body Boson systems.

In 3d, the dynamics of a system of $N$ interacting Bosons is governed by a
symmetric wave function which solves the $N-body$ Schr\"{o}dinger equation%
\begin{equation*}
i\partial _{t}\psi _{N}=(\triangle -V_{N})\psi _{N}\text{ in }\mathbb{R}%
^{3N+1}.
\end{equation*}%
Or written out explicitly, considering weak interactions and a condensation
initial, it is%
\begin{eqnarray}
i\partial _{t}\psi _{N} &=&\left( \sum_{i=1}^{N}\triangle _{x_{j}}-\frac{1}{%
N^{2}}\sum_{i<j<k}v_{N}(x_{i}-x_{j},x_{i}-x_{k})\right) \psi _{N}\text{ in }%
\mathbb{R}^{3N+1}  \label{Equation:3BodyPDE} \\
\psi _{N}|_{t=0} &=&\dprod\limits_{i=1}^{N}\phi _{0}(x_{i})\in L_{s}^{2}({%
\mathbb{R}}^{3N}),  \notag
\end{eqnarray}%
in the 3-body interaction case, and%
\begin{eqnarray}
i\partial _{t}\psi _{N} &=&\left( \sum_{i=1}^{N}\triangle _{x_{j}}-\frac{1}{N%
}\sum_{i<j}v_{N}(x_{i}-x_{j})\right) \psi _{N}\text{ in }\mathbb{R}^{3N+1}
\label{Equation:2BodyPDE} \\
\psi _{N}|_{t=0} &=&\dprod\limits_{i=1}^{N}\phi _{0}(x_{i})\in L_{s}^{2}({%
\mathbb{R}}^{3N}),  \notag
\end{eqnarray}%
in the 2-body interaction case. The reason that the initial $\psi
_{N}|_{t=0} $ is a tensor product for condensate is in \cite%
{Lieb1,Lieb2,Lieb3,Lieb4}.

However, as $N,$ the number of particles, increases, solving the $N-body$
Schr\"{o}dinger equation becomes unrealistic. To circumvent this difficulty,
we search for effective replacements or approximations which are
simultaneously easier to approach and accurate (in a suitable sense). We are
then led to the mean-field approximations in which the $N-body$ wave
function for the condensate is approximated by tensor products of a single
particle wave function satisfying an appropriate nonlinear Schr\"{o}dinger
equation.

Natural questions arise concerning the justification of the link between the
mean-field approximations and the actual many-body Hamiltonian evolutions.
Elgart, Erd\"{o}s, Schlein, and Yau~\cite%
{E-E-S-Y1,E-Y1,E-S-Y1,E-S-Y2,E-S-Y4,E-S-Y5, E-S-Y3} showed rigorously how
mean-field limits for Bosons can be extracted in the limit $N\rightarrow
\infty $ by using Bogoliubov-Born-Green-Kirkwood-Yvon (BBGKY) hierarchies
for reduced density matrices regarding equation \ref{Equation:2BodyPDE}. By
assuming a space-time estimate, Klainerman and Machedon gave another proof
of the uniqueness part of the argument in \cite{KlainermanAndMachedon}.
There was subsequent work by Kirkpatrick, Schlein, and Staffilani \cite%
{Kirpatrick}. In Chen and Pavlovi\'{c} \cite{TChenAndNP}, the reduced
density matrices approach of equation \ref{Equation:3BodyPDE} was studied.

\begin{remark}
It is currently an open question in this field to show directly that the
limit of the BBGKY hierarchy in 3d satisfies the space-time bound assumed by
Klainerman and Machedon. The answers in the 1d and 2d cases have been found.
See \cite{TChenAndNP} and \cite{Kirpatrick}.
\end{remark}

Concerning the convergence of the microscopic evolution towards the mean
field dynamics, Rodnianski and Schlein~provided estimates for the rate of
convergence in the case with Hartree dynamics by invoking the Fock space
formalism of equation \ref{Equation:2BodyPDE} in \cite{RodnianskiAndSchlein}.

Recently, following \cite{RodnianskiAndSchlein}, Grillakis, Machedon and
Margetis also studied the Fock space formalism of equation \ref%
{Equation:2BodyPDE} and found a second-order correction (GMM type
correction) inspired by Wu \cite{wuI, wuII} to the mean-field approximation 
\cite{GMM1, GMM2}. This paper aims to generalize their result to the Fock
space representation of equation \ref{Equation:3BodyPDE}.

The main motivation of this paper is to point out that when we apply the GMM
approximation to the Hamiltonian evolution of many-particle systems equipped
with 3-body interactions, the error between GMM approximation and the actual
many-body Hamiltonian evolutions can be controlled uniformly in time. (See
Knowles and Pickl \cite{KnowlesAndPickl} for another type of uniform error
bound.) We now state our problem mathematically. We will discuss the
difference bewteen the 2-body and 3-body cases in subsection \ref%
{subsection:comparison}.

\subsection{Fock Space Representation of The Problem}

First, we set up the Boson Fock space ${\mathcal{F}}$ following \cite{GMM1,
GMM2, RodnianskiAndSchlein}. One can also look at Folland \cite{Folland}.

\begin{definition}
\label{Def:FockSpace}The Hilbert space Boson Fock space ${\mathcal{F}}$
based on $L^{2}(\mathbb{R}^{3})$ contains vectors of the form ${\boldsymbol{%
\psi }}=(\psi _{0},\psi _{1}(x_{1}),\psi _{2}(x_{1},x_{2}),\cdots )$ where $%
\psi _{0}\in {\mathbb{C}}$ and $\psi _{n}\in L_{s}^{2}({\mathbb{R}}^{3n})$
are symmetric in $x_{1},\ldots ,x_{n}$. The Hilbert space structure of ${%
\mathcal{F}}$ is given by $\left( {\boldsymbol{\phi }},{\boldsymbol{\psi }}%
\right) =\sum_{n}\int \phi _{n}\overline{\psi _{n}}dx$.
\end{definition}

\begin{definition}
For $f\in L^{2}({\mathbb{R}}^{3}),$ we define the (unbounded, closed,
densely defined) creation operator $a^{\ast }(f):{\mathcal{F}}\rightarrow {%
\mathcal{F}}$ and annihilation operator $a(\bar{f}):{\mathcal{F}}\rightarrow 
{\mathcal{F}}$ by%
\begin{align*}
& \left( a^{\ast }(f)\psi _{n-1}\right) (x_{1},x_{2},\cdots ,x_{n})=\frac{1}{%
\sqrt{n}}\sum_{j=1}^{n}f(x_{j})\psi _{n-1}(x_{1},\cdots
,x_{j-1},x_{j+1},\cdots x_{n}), \\
& \left( a(\overline{f})\psi _{n+1}\right) (x_{1},x_{2},\cdots ,x_{n})=\sqrt{%
n+1}\int \psi _{(n+1)}(x,x_{1},\cdots ,x_{n})\overline{f}(x)dx.
\end{align*}%
The operator valued distributions $a_{x}^{\ast }$ and $a_{x}$ are defined by 
\begin{align*}
a^{\ast }(f)& =\int f(x)a_{x}^{\ast }dx, \\
a(\overline{f})& =\int \overline{f}(x)a_{x}dx.
\end{align*}%
These distributions satisfy the canonical commutation relations 
\begin{align}
\lbrack a_{x},a_{y}^{\ast }]& =\delta (x-y),
\label{Formula:CanonicalCommutation} \\
\lbrack a_{x},a_{y}]& =[a_{x}^{\ast },a_{y}^{\ast }]=0.  \notag
\end{align}
\end{definition}

Define the vacuum state $\Omega \in {\mathcal{F}}$ and the skew-Hermitian
unbounded operator $A$ by 
\begin{eqnarray}
\Omega &=&(1,0,0,\cdots )  \notag \\
A(\phi ) &=&a(\overline{\phi })-a^{\ast }(\phi ).  \label{Def:A}
\end{eqnarray}%
It is easy to check that%
\begin{equation*}
e^{-\sqrt{N}A(\phi _{0})}\Omega =e^{-N\Vert \phi \Vert ^{2}/2}\left(
1,\cdots ,\left( \frac{N^{n}}{n!}\right) ^{1/2}\phi _{0}(x_{1})\cdots \phi
_{0}(x_{n}),\cdots \right) .
\end{equation*}%
This renders the Fock space analogue of the initial data 
\begin{equation*}
\psi _{N}|_{t=0}=\dprod\limits_{i=1}^{N}\phi _{0}(x_{i})
\end{equation*}%
in equation \ref{Equation:3BodyPDE}.

In correspondence to equation \ref{Equation:3BodyPDE}, we consider the Fock
space $N$ particle Hamiltonian $H_{N}:{\mathcal{F}}\rightarrow {\mathcal{F}}$
with the 3-body interaction potential 
\begin{eqnarray}
H_{N} &=&\int a_{x}^{\ast }\Delta a_{x}dx-\frac{1}{6N^{2}}\int
v_{3}(x-y,x-z)a_{x}^{\ast }a_{y}^{\ast }a_{z}^{\ast }a_{x}a_{y}a_{z}dxdydz
\label{Formula:3BodyHamiltonian} \\
&=&H_{0}-\frac{1}{6N^{2}}V,  \notag
\end{eqnarray}%
where the triple interaction $v_{3}$ is assumed to be symmetric in $x$, $y$,
and $z$. As shown in \cite{TChenAndNP}, every translation invariant 3-body
potential can be written in the form $v_{3}(x-y,x-z)$.

\begin{remark}
Since we use $\triangle $ instead of $-\triangle $ in formula \ref%
{Formula:3BodyHamiltonian}, nonnegative $v_{3}$ corresponds to the
defocusing case.
\end{remark}

Whence the Fock space representation of equation \ref{Equation:3BodyPDE} is
the Hamiltonian evolution 
\begin{equation*}
e^{itH_{N}}e^{-\sqrt{N}A(\phi _{0})}\Omega .
\end{equation*}

Let the one-particle wave function $\phi (t,x)$ solve the quintic Hartree
equation 
\begin{equation}
i\frac{\partial }{\partial t}\phi +\triangle \phi -\frac{1}{2}\phi \int
v_{3}(x-y,x-z)\left\vert \phi (y)\right\vert ^{2}\left\vert \phi
(z)\right\vert ^{2}dydz=0  \label{Equation:HartreeInIntroduction}
\end{equation}%
subject to the initial condition $\phi (0,x)=\phi _{0}(x)$, then the mean
field approximation for $e^{itH_{N}}e^{-\sqrt{N}A(\phi _{0})}\Omega $ is the
tensor product of $\phi (t,x),$ or%
\begin{equation*}
{\boldsymbol{\psi }}_{MeanField}=e^{-\sqrt{N}A(\phi (t,\cdot ))}\Omega
\end{equation*}%
to be specific. A derivation of equation \ref{Equation:HartreeInIntroduction}
is given in Section \ref{section:ProofOf2ndOrderTheorem}.

When the Hamiltonian $H_{N}$ is the Hamiltonian subject to the two body
interaction 
\begin{eqnarray*}
H_{N,2} &=&\int a_{x}^{\ast }\Delta a_{x}dx-\frac{1}{2N}\int
v_{2}(x-y)a_{x}^{\ast }a_{y}^{\ast }a_{x}a_{y}dxdy \\
&=&H_{0}-\frac{1}{N}V_{2}
\end{eqnarray*}%
which is the Fock space formalism of equation \ref{Equation:2BodyPDE},
Rodnianski and Schlein~\cite{RodnianskiAndSchlein} derived a cubic Hartree
equation for $\phi (t,x)$ (equation \ref{Equation:CubicHartreeEquation}) and
showed that the mean field approximation works (under suitable assumptions
on $v$) in the sense that 
\begin{align*}
& \frac{1}{N}\Vert \left( e^{itH_{N,2}}{\boldsymbol{\psi }}%
_{0},\,a_{y}^{\ast }a_{x}e^{itH_{N,2}}{\boldsymbol{\psi }}_{0}\right)
-\left( e^{-\sqrt{N}A(\phi (t,\cdot ))}\Omega ,\,a_{y}^{\ast }a_{x}e^{-\sqrt{%
N}A(\phi (t,\cdot ))}\Omega \right) \Vert _{Tr} \\
& =O(\frac{e^{Ct}}{N})\qquad N\rightarrow \infty ~;
\end{align*}%
where $\Vert \cdot \Vert _{Tr}$ stands for the trace norm in $x\in {\mathbb{R%
}}^{3}$ and $y\in {\mathbb{R}}^{3},$ and ${\boldsymbol{\psi }}_{0}=e^{-\sqrt{%
N}A(\phi _{0})}\Omega $. For the precise statement of the problem and
details of the proof, see Theorem 3.1 of Rodnianski and Schlein \cite%
{RodnianskiAndSchlein}. Later, in \cite{GMM1, GMM2}, Grillakis, Machedon and
Margetis introduced a second-order correction (GMM type correction) to the
mean field approximation of $e^{itH_{N,2}}e^{-\sqrt{N}A(\phi _{0})}\Omega $
which greatly improved the error.

Instead of delving into the results in \cite{GMM1, GMM2}, we state our main
theorems first. This makes it easier to compare our results with the ones in 
\cite{GMM1, GMM2}.

In the main theorem of this paper, we consider the defocusing case:%
\begin{equation*}
v_{3}(x-y,x-z)=v(x-y,x-z)
\end{equation*}%
where%
\begin{equation}
v(x-y,x-z)=v_{0}(x-y)v_{0}(x-z)+v_{0}(x-y)v_{0}(y-z)+v_{0}(x-z)v_{0}(y-z)
\label{Formula:DefinitonOfv-epsilon}
\end{equation}%
built of a nonnegative regular potential $v_{0}$ which decays fast enough
away from the origin and has the property that 
\begin{equation*}
v_{0}(x)=v_{0}(-x).
\end{equation*}

We remark that our main theorems also work when $v_{0}$ has a proper
singularity at the origin. To be specific, if for some $\varepsilon \in (0,%
\frac{1}{2})$%
\begin{equation}
v_{0}(x)=\frac{\chi (\left\vert x\right\vert )}{\left\vert x\right\vert
^{1-\varepsilon }},\text{ or }G_{2+\varepsilon }(x)
\label{def:singular potential}
\end{equation}%
where $\chi \in C_{0}^{\infty }(\mathbb{R}^{+}\mathbb{\cup \{}0\mathbb{\}})$
is nonnegative and decreasing and $G_{\alpha }$ the kernel of Bessel
potential, then Theorems \ref{THM:MainTheorem} holds. Though we currently do
not know the physical meaning for such potentials if $\varepsilon \neq 0$,
we would like to understand the analysis when singularities appear since the
derivation of the quintic NLS uses an interacion which goes to a delta
function when $N\rightarrow \infty .$ Due to the technicality of treating
the singularities, we restrict to the case of smooth potentials so that the
differences between the 2-body and 3-body interactions are easier to see.

\begin{remark}
\label{Remark:Simplicity}For simplicity, let us write $A(\phi )$ as $A$, $%
A(\phi (t,\cdot ))$ as $A(t)$, $v_{3}(x-y,x-z)$ as $v_{3,1-2,1-3},$ and $%
\phi (y)$ as $\phi _{2}$ etc.
\end{remark}

\subsection{Statement of The Main Theorems}

\begin{theorem}
\label{THM:MainTheorem}Assume the defocusing case $v_{3}(x-y,x-z)=v(x-y,x-z)$
where $v(x-y,x-z)$ is defined in formula \ref{Formula:DefinitonOfv-epsilon}.

If $\phi _{0},$ the initial data, satisfies

(i) finite mass: 
\begin{equation*}
\left\Vert \phi _{0}\right\Vert _{L_{x}^{2}}=1,
\end{equation*}

(ii) finite energy:%
\begin{eqnarray*}
E_{0} &=&\frac{1}{2}\int \left\vert \nabla \phi _{0}\right\vert ^{2}dx+\frac{%
1}{6}\int v(x-y,x-z)\left\vert \phi _{0}(x)\right\vert ^{2}\left\vert \phi
_{0}(y)\right\vert ^{2}\left\vert \phi _{0}(z)\right\vert ^{2}dxdydz \\
&\leqslant &C_{1},
\end{eqnarray*}

(iii) finite variance:%
\begin{equation*}
\left\Vert \left\vert \cdot \right\vert \phi _{0}\right\Vert
_{L_{x}^{2}}\leqslant C_{2},
\end{equation*}%
then based on the tensor product approximation (mean-field), we can
construct ${\boldsymbol{\psi }}_{GMM}$, a GMM type approximation explained
in Theorem \ref{THM:2ndOrderCorrection}, to the Hamiltonian evolution $%
e^{itH_{N}}e^{-\sqrt{N}A(\phi _{0})}\Omega $ for $H_{N}$ defined in formula %
\ref{Formula:3BodyHamiltonian}, such that the following uniform in time
error estimate holds%
\begin{equation*}
\Vert {\boldsymbol{\psi }}_{GMM}-e^{itH_{N}}e^{-\sqrt{N}A(\phi _{0})}\Omega
\,\Vert _{{\mathcal{F}}}\leqslant \frac{C}{\sqrt{N}}
\end{equation*}%
where ${\mathcal{F}}$ is the Boson Fock space defined in Definition \ref%
{Def:FockSpace} and $C$ depends only on $v$, $C_{1}$ and $C_{2}$.
\end{theorem}

\begin{remark}
The construction of ${\boldsymbol{\psi }}_{GMM}$ does not require $v_{3}$ to
have a definite sign. However, we take positive sign in Theorem \ref%
{THM:MainTheorem} because it leads to a defocusing Hartree equation whose
global behavior is controllable.
\end{remark}

We prove Theorem \ref{THM:MainTheorem} via Theorems \ref%
{THM:2ndOrderCorrection} and \ref{THM:MainTheoremErrorEstimate} stated
below. They deal with the construction of ${\boldsymbol{\psi }}_{GMM}$ and
the error estimate separately. However, it is worth pointing out that
Theorem \ref{THM:MainTheorem} is a special case of Theorems \ref%
{THM:2ndOrderCorrection} and \ref{THM:MainTheoremErrorEstimate}, which apply
to a more general setting beyond initial data of the form $e^{-\sqrt{N}%
A(\phi _{0})}\Omega $.

\begin{theorem}
\label{THM:2ndOrderCorrection}Let $\phi $ be a sufficiently smooth solution
of the quintic Hartree equation 
\begin{equation}
i\frac{\partial }{\partial t}\phi +\triangle \phi -\frac{1}{2}\phi \int
v_{3}(x-y,x-z)\left\vert \phi (y)\right\vert ^{2}\left\vert \phi
(z)\right\vert ^{2}dydz=0  \label{Equation:QuinticHartreeEquation}
\end{equation}%
with initial data $\phi _{0}$ and the 3-body interaction potential $v_{3}$
being symmetric in $x$, $y$, and $z$. Assume the following:

(1) Let a complex kernel $k(t,x,y)\in L_{s}^{2}(dxdy)$ for almost all $t$,
solve the equation 
\begin{equation}
iu_{t}+ug^{T}+gu-(I+p)m=\left( ip_{t}+[g,p]+u\overline{m}\right) \left(
I+p\right) ^{-1}u,  \label{Equation:EquationOfk}
\end{equation}%
with 
\begin{align*}
u(t,x,y)& :=\sinh (k):=k+\frac{1}{3!}k\overline{k}k+\ldots ~, \\
\cosh (k)& :=I+p(t,x,y):=\delta (x-y)+\frac{1}{2!}k\overline{k}+\ldots ~, \\
g(t,x,y)& :=-\triangle \delta (x-y)+\left( \int v_{3}(x-y,x-z)\left\vert
\phi (z)\right\vert ^{2}dz\right) \overline{\phi }(x)\overline{\phi }(y) \\
& +\frac{1}{2}\left( \int v_{3}(x-y,x-z)\left\vert \phi (y)\right\vert
^{2}\left\vert \phi (z)\right\vert ^{2}dydz\right) \delta (x-y), \\
m(t,x,y)& :=-\left( \int v_{3}(x-y,x-z)\left\vert \phi (z)\right\vert
^{2}dz\right) \overline{\phi }(x)\overline{\phi }(y),
\end{align*}%
where the products $ug^{T}$, $k\overline{k}$ etc. stand for compositions of
operators.

(2) For $V$ defined as in formula \ref{Formula:3BodyHamiltonian}, the
functions 
\begin{equation*}
\Vert e^{B}Ve^{-B}\Omega \Vert _{{\mathcal{F}}},\Vert e^{B}[A,V]e^{-B}\Omega
\Vert _{{\mathcal{F}}},\Vert e^{B}[A,[A,V]]e^{-B}\Omega \Vert _{{\mathcal{F}}%
},\text{and }\Vert e^{B}[A,[A,[A,V]]]e^{-B}\Omega \Vert _{{\mathcal{F}}}
\end{equation*}%
are locally integrable in time, where 
\begin{equation}
B(t):=\frac{1}{2}\int \left( k(t,x,y)a_{x}a_{y}-\overline{k}%
(t,x,y)a_{x}^{\ast }a_{y}^{\ast }\right) dxdy.  \label{Formula:DefinitionOfB}
\end{equation}

(3) $\int d(t,x,x)\ dx$ is also locally integrable in time, where 
\begin{align}
d(t,x,y):=& \left( i\sinh (k)_{t}+\sinh (k)g^{T}+g\sinh (k)\right) \overline{%
\sinh (k)}  \label{Formula:DefinitionOFd} \\
-& \left( i\cosh (k)_{t}+[g,\cosh (k)]\right) \cosh (k)  \notag \\
-& \sinh (k)\overline{m}\cosh (k)-\cosh (k)m\overline{\sinh (k)}.  \notag
\end{align}%
Then we define%
\begin{equation*}
{\boldsymbol{\psi }}_{GMM}:=e^{-\sqrt{N}A(\phi (t,\cdot
))}e^{-B(t)}e^{-i\int_{0}^{t}(N\chi _{0}(s)+\chi _{1}(s))ds}\Omega
\end{equation*}%
where%
\begin{eqnarray*}
\chi _{0}(t) &:&=-\frac{1}{3}\int v_{3}(x-y,x-z)\left\vert \phi
(x)\right\vert ^{2}\left\vert \phi (y)\right\vert ^{2}\left\vert \phi
(z)\right\vert ^{2}dxdydz, \\
\chi _{1}(t) &:&=\frac{1}{2}\int d(t,x,x)dx.
\end{eqnarray*}%
This definition of ${\boldsymbol{\psi }}_{GMM}$ yields the error estimate%
\begin{align*}
& \Vert {\boldsymbol{\psi }}_{GMM}-e^{itH_{N}}e^{-\sqrt{N}A(\phi
_{0})}e^{-B(0)}\Omega \,\Vert _{{\mathcal{F}}} \\
& \leqslant \frac{\int_{0}^{t}\Vert e^{B}Ve^{-B}\Omega \Vert _{{\mathcal{F}}%
}ds}{6N^{2}}+\frac{\int_{0}^{t}\Vert e^{B}[A,V]e^{-B}\Omega \Vert _{{%
\mathcal{F}}}ds}{6N^{\frac{3}{2}}} \\
& +\frac{\int_{0}^{t}\Vert e^{B}[A,[A,V]]e^{-B}\Omega \Vert _{{\mathcal{F}}%
}ds}{12N}+\frac{\int_{0}^{t}\Vert e^{B}[A,[A,[A,V]]]e^{-B}\Omega \Vert _{{%
\mathcal{F}}}ds}{36N^{\frac{1}{2}}}.
\end{align*}
\end{theorem}

\begin{theorem}
\label{THM:MainTheoremErrorEstimate}Assume $v_{3}(x-y,x-z)=v(x-y,x-z)$ i.e.
equation \ref{Equation:QuinticHartreeEquation} becomes%
\begin{equation}
i\frac{\partial }{\partial t}\phi +\triangle \phi -\frac{1}{2}\phi \int
v(x-y,x-z)\left\vert \phi (y)\right\vert ^{2}\left\vert \phi (z)\right\vert
^{2}dydz=0.  \label{Equation:DefocusingHartreeEquation}
\end{equation}%
If $\phi _{0},$ the initial data of quintic Hartree equation \ref%
{Equation:DefocusingHartreeEquation}, satisfies (i), (ii), and (iii), then
the hypotheses in Theorem \ref{THM:2ndOrderCorrection} are satisfied
globally in time. Moreover, we have the error estimate uniformly in time that%
\begin{equation*}
\Vert {\boldsymbol{\psi }}_{GMM}-e^{itH_{N}}e^{-\sqrt{N}A(\phi
_{0})}e^{-B(0)}\Omega \,\Vert _{{\mathcal{F}}}\leqslant \frac{C}{\sqrt{N}}
\end{equation*}%
where $C$ depends only on $v$, $C_{1}$, $C_{2}$ and $\left\Vert u(0,\cdot
,\cdot )\right\Vert _{L_{(x,y)}^{2}}$.
\end{theorem}

We deduce Theorem \ref{THM:MainTheorem} from Theorems \ref%
{THM:2ndOrderCorrection} and \ref{THM:MainTheoremErrorEstimate} by setting 
\begin{equation*}
k(0,x,y)=0.
\end{equation*}

The proof of Theorem \ref{THM:MainTheoremErrorEstimate} relies on the
following theorem regarding the long time behavior of the solution to the
Hartree equation.

\begin{theorem}
\label{THM:TheLongTimeBehaviorOfHartree}If $\phi $ solve the Hartree
equation \ref{Equation:DefocusingHartreeEquation} subject to (i), (ii), and
(iii), then%
\begin{equation*}
\left\Vert \phi \right\Vert _{L_{x}^{6}}\leqslant \frac{C}{t},\text{ }for%
\text{ }t\geqslant 1,
\end{equation*}%
where $C$ is a function of $v$, $C_{1}$ and $C_{2}$ only.
\end{theorem}

\subsection{Comparison with Results in \protect\cite{GMM1, GMM2}\label%
{subsection:comparison}}

In Theorem \ref{THM:2ndOrderCorrection}, if we change $H_{N}$ to $H_{N,2}$,
and equation \ref{Equation:QuinticHartreeEquation} to%
\begin{equation}
i\frac{\partial }{\partial t}\phi +\triangle \phi -\phi \int
v_{2}(x-y)\left\vert \phi (y)\right\vert ^{2}dy=0,
\label{Equation:CubicHartreeEquation}
\end{equation}%
and we let%
\begin{align*}
g(t,x,y)& =-\Delta \delta (x-y)+v_{2}(x-y)\phi (t,x)\overline{\phi }%
(t,y)+(v_{2}\ast |\phi |^{2})(t,x)\delta (x-y), \\
m(t,x,y)& =-v_{2}(x-y)\overline{\phi }(t,x)\overline{\phi }(t,y), \\
\chi _{0}(t)& =-\frac{1}{2}\int v_{2}(x-y)\left\vert \phi (x)\right\vert
^{2}\left\vert \phi (y)\right\vert ^{2}dxdy,
\end{align*}%
then Theorem 1 in \cite{GMM1} reads%
\begin{align*}
& \Vert e^{-\sqrt{N}A(t)}e^{-B(t)}e^{-i\int_{0}^{t}(N\chi _{0}(s)+\chi
_{1}(s))ds}\Omega -e^{itH_{N,2}}e^{-\sqrt{N}A(0)}e^{-B(0)}\Omega \,\Vert _{{%
\mathcal{F}}} \\
& \leqslant \frac{\int_{0}^{t}\Vert e^{B}V_{2}e^{-B}\Omega \Vert _{{\mathcal{%
F}}}ds}{N}+\frac{\int_{0}^{t}\Vert e^{B}[A,V_{2}]e^{-B}\Omega \Vert _{{%
\mathcal{F}}}ds}{N^{\frac{1}{2}}}.
\end{align*}%
If $v_{2}(x)=\frac{\chi (\left\vert x\right\vert )}{\left\vert x\right\vert }
$, the above error estimate becomes%
\begin{align}
& \Vert e^{-\sqrt{N}A(t)}e^{-B(t)}e^{-i\int_{0}^{t}(N\chi _{0}(s)+\chi
_{1}(s))ds}\Omega -e^{itH_{N,2}}e^{-\sqrt{N}A(0)}e^{-B(0)}\Omega \,\Vert _{{%
\mathcal{F}}}  \label{estimate:mateilongtimeestimate} \\
& \leqslant \frac{C(1+t)^{\frac{1}{2}}}{\sqrt{N}}  \notag
\end{align}%
in \cite{GMM2}.

Compared with the above long time estimate, Theorem \ref{THM:MainTheorem}
demonstrates that there is a substantial difference between the 3-body
interaction case and 2-body interaction. Technically speaking, the main
difference between the 2-body and 3-body interactions lies in their error
terms. Though the analysis is more involved even if we assume smooth
potential and the formulas are considerably longer, the more complicated
error terms in the 3-body interaction case in fact allow more room to play.
On the one hand, an error term in the 3-body case carries at least a pair of 
$u$, $p$ or $\phi $ which satisfy some Schr\"{o}dinger equations, for
instance, the term%
\begin{equation*}
\left\Vert v_{3}(x_{1}-y_{1},x_{1}-z_{1})\bar{u}(t,x_{2},x_{1})\bar{u}%
(t,y_{1},z_{1})\right\Vert _{L_{t}^{1}L^{2}}
\end{equation*}%
in formula \ref{formula:Worst3-bodyErrorTerm}, which can be estimated by
Lemma \ref{Lemma:KeyLemmaForErrorTerms}. A typical error term in the 2-body
case can carry only one term of $u$, $p$ or $\phi $, for example, the term%
\begin{equation*}
\left\Vert v_{2}(x_{1}-y_{1})u(t,y_{1},x_{1})\right\Vert
_{L_{t}^{1}([0,T])L^{2}}
\end{equation*}%
implicitly inside formula 47 in \cite{GMM1}. On the other hand, the error
estimate in the construction of the second order correction involves $%
L_{t}^{1}$ and we have no $L_{t}^{1}$ dispersive estimates for the Schr\"{o}%
dinger equation. Therefore, due to the endpoint Strichartz estimates \cite%
{KeelAndTao}, we can construct a $L_{t}^{1}(\mathbb{R}^{+})$ estimate for
the 3-body case which is Lemma \ref{Lemma:KeyLemmaForErrorTerms}, without
having the $t^{\frac{1}{2}}$ in the 2-body case which is necessary to apply
the $L_{t}^{2}$ Kato estimate in \cite{GMM1, GMM2}. Or in other words, we do
Cauchy-Schwarz in time differently.

For the reason stated above, one can not employ the 3-body case error
estimate in the 2-body case. Furthermore, the tool of error estimate in the
2-body case \cite{GMM1, GMM2}, does not apply to the 3-body case, no matter
if $v_{3}$ is regular or singular like formula \ref{def:singular potential}.

\subsection{Outline of The Paper}

We prove Theorem \ref{THM:2ndOrderCorrection} in Section \ref%
{section:ProofOf2ndOrderTheorem}. The derivation of equations \ref%
{Equation:QuinticHartreeEquation} and \ref{Equation:EquationOfk} is also
included there. Section \ref{section:ProofOf2ndOrderTheorem} is similar to
Sections 3-5 in \cite{GMM1}. They share the same basic ideas though the
computation in this paper is more complicated. Therefore we refer the
readers to Sections 3-5 in \cite{GMM1} for details of the infinitesimal
metaplectic representation of symplectic matrices, the main tool of Section %
\ref{section:ProofOf2ndOrderTheorem}, and the rigorous definition of $e^{B}$
etc.

Sections \ref{section:EstimatesFor u} and \ref{section:ErrorEstimates} are
devoted to the proof of Theorem \ref{THM:MainTheoremErrorEstimate}. (1) and
(3) are verified in Section \ref{section:EstimatesFor u}. Error estimates
are sorted out in Section \ref{section:ErrorEstimates}. Like Section 3-4 in 
\cite{GMM2}, we first prove apriori estimates for $u$ which solves equation %
\ref{Equation:EquationOfk}, which are needed in the error estimates, Section %
\ref{section:ErrorEstimates}. Because the general scheme has been set up in 
\cite{GMM1,GMM2}, the details of some basic lemmas and theorems are omitted
e.g. Theorem \ref{THM:ExistenceOfU}\ in Section \ref{section:EstimatesFor u}%
. As mentioned before, our error estimates are done very differently from
the corresponding ones in \cite{GMM2}. The $L_{t}^{1}$ estimate, Lemma \ref%
{Lemma:KeyLemmaForErrorTerms}, is our key lemma for error estimates.

For smoothness of the presentation of Theorems \ref{THM:2ndOrderCorrection}
and \ref{THM:MainTheoremErrorEstimate}, we postpone the proof of Theorem \ref%
{THM:TheLongTimeBehaviorOfHartree} to Section \ref{section:AppendixHartree}.

\section{The Derivation of 2nd Order Corrections / Proof of Theorem \protect
\ref{THM:2ndOrderCorrection}}

\label{section:ProofOf2ndOrderTheorem}

\subsection{Derivation of The Quintic Hartree equation}

We first derive the quintic Hartree equation \ref%
{Equation:QuinticHartreeEquation} for the one-particle wave function $\phi $
as needed in Theorem \ref{THM:2ndOrderCorrection}.

\begin{lemma}
\label{Lemma:Commutators} The following commutating relations hold, where $A$
denotes $A(\phi )$, and $A$ , $V$ are defined by formulas \ref{Def:A} and %
\ref{Formula:3BodyHamiltonian} : 
\begin{eqnarray*}
&&[A,V] \\
&=&3\int v_{3}(x-y,x-z)(\overline{\phi }(x)a_{y}^{\ast }a_{z}^{\ast
}a_{x}a_{y}a_{z}+\phi (x)a_{x}^{\ast }a_{y}^{\ast }a_{z}^{\ast
}a_{y}a_{z})dxdydz
\end{eqnarray*}%
\begin{eqnarray*}
&&[A,[A,V]] \\
&=&6\int v_{3}(x-y,x-z)(\overline{\phi }(x)\overline{\phi }(y)a_{z}^{\ast
}a_{x}a_{y}a_{z}+2\phi (x)\overline{\phi }(y)a_{x}^{\ast }a_{z}^{\ast
}a_{y}a_{z} \\
&&+\phi (x)\phi (y)a_{x}^{\ast }a_{y}^{\ast }a_{z}^{\ast }a_{z})dxdydz \\
&&+6\int v_{3}(x-y,x-z)\left\vert \phi (x)\right\vert ^{2}a_{y}^{\ast
}a_{z}^{\ast }a_{y}a_{z}dxdydz
\end{eqnarray*}%
\begin{eqnarray*}
&&[A,[A,[A,V]]] \\
&=&36\int v_{3}(x-y,x-z)\left\vert \phi (x)\right\vert ^{2}(\overline{\phi }%
(y)a_{z}^{\ast }a_{y}a_{z}+\phi (y)a_{y}^{\ast }a_{z}^{\ast }a_{z})dxdydz \\
&&+6\int v_{3}(x-y,x-z)(\overline{\phi }(x)\overline{\phi }(y)\overline{\phi 
}(z)a_{x}a_{y}a_{z}+\phi (x)\phi (y)\phi (z)a_{x}^{\ast }a_{y}^{\ast
}a_{z}^{\ast })dxdydz \\
&&+18\int v_{3}(x-y,x-z)(\overline{\phi }(x)\overline{\phi }(y)\phi
(z)a_{z}^{\ast }a_{x}a_{y}+\phi (x)\phi (y)\overline{\phi }(z)a_{x}^{\ast
}a_{y}^{\ast }a_{z})dxdydz
\end{eqnarray*}%
\begin{eqnarray*}
&&[A,[A,[A,[A,V]]]] \\
&=&72\int v_{3}(x-y,x-z)\left\vert \phi (x)\right\vert ^{2}(\overline{\phi }%
(y)\overline{\phi }(z)a_{y}a_{z}+\phi (y)\phi (z)a_{y}^{\ast }a_{z}^{\ast
})dxdydz \\
&&+144\int v_{3}(x-y,x-z)\left\vert \phi (x)\right\vert ^{2}\overline{\phi }%
(y)\phi (z)a_{z}^{\ast }a_{y}dxdydz \\
&&+72\int v_{3}(x-y,x-z)\left\vert \phi (x)\right\vert ^{2}\left\vert \phi
(y)\right\vert ^{2}a_{z}^{\ast }a_{z}dxdydz
\end{eqnarray*}%
\begin{eqnarray*}
&&[A,[A,[A,[A,[A,V]]]]] \\
&=&360\int v_{3}(x-y,x-z)\left\vert \phi (x)\right\vert ^{2}\left\vert \phi
(y)\right\vert ^{2}(\overline{\phi }(z)a_{z}+\phi (z)a_{z}^{\ast })dxdydz
\end{eqnarray*}%
\begin{eqnarray*}
&&[A,[A,[A,[A,[A,[A,V]]]]]] \\
&=&720\int v_{3}(x-y,x-z)\left\vert \phi (x)\right\vert ^{2}\left\vert \phi
(y)\right\vert ^{2}\left\vert \phi (z)\right\vert ^{2}dxdydz
\end{eqnarray*}
\end{lemma}

\begin{proof}
This is a direct calculation using the canonical commutation relation \ref%
{Formula:CanonicalCommutation}.
\end{proof}

Now, we write $\Psi _{0}(t)=e^{\sqrt{N}A(t)}e^{itH_{N}}e^{-\sqrt{N}%
A(0)}e^{-B(0)}\Omega $ for which we carry out the calculation in the spirit
of equation~3.7 in Rodnianski and Schlein \cite{RodnianskiAndSchlein}.

\begin{proposition}
Let $\phi $ solve the Hartree equation 
\begin{equation*}
i\frac{\partial }{\partial t}\phi +\triangle \phi -\frac{1}{2}\phi \int
v_{3}(x-y,x-z)\left\vert \phi (y)\right\vert ^{2}\left\vert \phi
(z)\right\vert ^{2}dydz=0
\end{equation*}%
then $\Psi _{0}(t)$ satisfies 
\begin{align}
& \frac{1}{i}\frac{\partial }{\partial t}\Psi _{0}(t)=\bigg(H_{0}-\frac{1}{4!%
}\frac{1}{6}[A,[A,[A,[A,V]]]]-\frac{1}{6}N^{-2}V-\frac{1}{6}N^{-3/2}[A,V]
\label{Equation:r-s equation} \\
& -\frac{1}{12}N^{-1}[A,[A,V]]-\frac{1}{36}N^{-\frac{1}{2}}[A,[A,[A,V]]] 
\notag \\
& +\frac{N}{3}\int v_{3}(x-y,x-z)\left\vert \phi (x)\right\vert
^{2}\left\vert \phi (y)\right\vert ^{2}\left\vert \phi (z)\right\vert
^{2}dxdydz\bigg)\Psi _{0}(t)~.  \notag
\end{align}
\end{proposition}

\begin{proof}
Applying the formulas 
\begin{equation*}
\left( \frac{\partial }{\partial t}e^{C(t)}\right) \left( e^{-C(t)}\right) =%
\dot{C}+\frac{1}{2!}[C,\dot{C}]+\frac{1}{3!}\left[ C,[C,\dot{C}]\right]
+\ldots
\end{equation*}%
\begin{equation*}
e^{C}He^{-C}=H+[C,H]+\frac{1}{2!}\left[ C,[C,H]\right] +\ldots ~.
\end{equation*}%
to $C=\sqrt{N}A$ and $H=H_{N},$ we obtain

\begin{equation}
\frac{1}{i}\frac{\partial }{\partial t}\Psi _{0}(t)=L_{0}\Psi _{0}~,
\label{Equation:r-s}
\end{equation}%
where 
\begin{align*}
& L_{0}=\frac{1}{i}\left( \frac{\partial }{\partial t}e^{\sqrt{N}%
A(t)}\right) e^{-\sqrt{N}A(t)}+e^{\sqrt{N}A(t)}H_{N}e^{-\sqrt{N}A(t)} \\
=& \frac{1}{i}\left( N^{1/2}\dot{A}+\frac{N}{2}[A,\dot{A}]\right)
+H_{0}+N^{1/2}[A,H_{0}]+\frac{N}{2!}[A,[A,H_{0}]]-\frac{1}{6}\bigg(%
N^{-2}V+N^{-3/2}[A,V] \\
& +\frac{N^{-1}}{2!}[A,[A,V]]+\frac{N^{-\frac{1}{2}}}{3!}[A,[A,[A,V]]] \\
& +\frac{1}{4!}[A,[A,[A,[A,V]]]]+\frac{N^{\frac{1}{2}}}{5!}%
[A,[A,[A,[A,[A,V]]]]] \\
& +\frac{N}{6!}[A,[A,[A,[A,[A,[A,V]]]]]]\bigg)\text{ }.
\end{align*}

The Hartree equation \ref{Equation:QuinticHartreeEquation} is equivalent to
setting terms of order $\sqrt{N}$ to zero i.e.%
\begin{equation*}
\frac{1}{i}\dot{A}+[A,H_{0}]-\frac{1}{6}\frac{1}{5!}[A,[A,[A,[A,[A,V]]]]]=0~.
\end{equation*}%
Or more explicitly, the above equation is 
\begin{equation*}
a(\overline{i\phi _{t}}+\overline{\triangle \phi }-\frac{1}{2}\overline{\phi 
}\int v_{3,1-2,1-3}\left\vert \phi _{2}\right\vert ^{2}\left\vert \phi
_{3}\right\vert ^{2}dydz)+a^{\ast }(i\phi _{t}+\triangle \phi -\frac{1}{2}%
\phi \int v_{3,1-2,1-3}\left\vert \phi _{2}\right\vert ^{2}\left\vert \phi
_{3}\right\vert ^{2}dydz)=0~,
\end{equation*}%
via Lemma \ref{Lemma:Commutators} and the fact that $[\Delta
_{x}a_{x},a_{y}^{\ast }]=(\Delta \delta )(x-y)$.

Thus 
\begin{equation*}
\frac{1}{i}[A,\dot{A}]+[A,[A,H_{0}]]-\frac{1}{5!}\frac{1}{6}%
[A,[A,[A,[A,[A,[A,V]]]]]]=0~,
\end{equation*}%
i.e. equation \ref{Equation:r-s} simplifies to 
\begin{align*}
& \frac{1}{i}\frac{\partial }{\partial t}\Psi _{0}(t)=\bigg(H_{0}-\frac{1}{4!%
}\frac{1}{6}[A,[A,[A,[A,V]]]]-\frac{1}{6}N^{-2}V-\frac{1}{6}N^{-3/2}[A,V] \\
& -\frac{1}{12}N^{-1}[A,[A,V]]-\frac{1}{36}N^{-\frac{1}{2}}[A,[A,[A,V]]] \\
& +\frac{N}{3}\int v_{3,1-2,1-3}\left\vert \phi _{1}\right\vert
^{2}\left\vert \phi _{2}\right\vert ^{2}\left\vert \phi _{3}\right\vert
^{2}dxdydz\bigg)\Psi _{0}(t)~.
\end{align*}%
which is equation \ref{Equation:r-s equation}.
\end{proof}

Because $\int v_{3,1-2,1-3}\left\vert \phi _{1}\right\vert ^{2}\left\vert
\phi _{2}\right\vert ^{2}\left\vert \phi _{3}\right\vert ^{2}dxdydz$ only
contributes a phase when $\phi _{0}$ is sufficiently smooth, we write 
\begin{equation*}
\frac{N}{3}\int v_{3}(x-y,x-z)\left\vert \phi (x)\right\vert ^{2}\left\vert
\phi (y)\right\vert ^{2}\left\vert \phi (z)\right\vert ^{2}dxdydz:=-N\chi
_{0}~.
\end{equation*}%
Then the first two terms on the right-hand side of equation \ref%
{Equation:r-s equation} are the main ones we need to consider, since the
next four terms are at most $O\left( 1/\sqrt{N}\right) $.

In order to kill the terms involving "only creation operators" i.e. $%
a_{x}^{\ast }a_{y}^{\ast }$ in $\frac{1}{4!}\frac{1}{6}[A,[A,[A,[A,V]]]]$,
we introduce $B$ (see \ref{Formula:DefinitionOfB}) and denote 
\begin{equation}
\Psi =e^{B}\Psi _{0}=e^{B}e^{\sqrt{N}A(t)}e^{itH_{N}}e^{-\sqrt{N}%
A(0)}e^{-B(0)}\Omega .  \notag
\end{equation}%
Hence we have 
\begin{equation}
\frac{1}{i}\frac{\partial }{\partial t}\Psi =L\Psi ~,  \notag
\end{equation}%
where

\begin{align*}
& L=\frac{1}{i}\left( \frac{\partial }{\partial t}e^{B}\right)
e^{-B}+e^{B}L_{0}e^{-B} \\
& =L_{Q}-\frac{1}{6}N^{-2}e^{B}Ve^{-B}-\frac{1}{6}N^{-3/2}e^{B}[A,V]e^{-B}-%
\frac{1}{12}N^{-1}e^{B}[A,[A,V]]e^{-B} \\
& -\frac{1}{36}N^{-\frac{1}{2}}e^{B}[A,[A,[A,V]]]e^{-B}-N\chi _{0}~,
\end{align*}%
and the quadratic terms 
\begin{equation}
L_{Q}=\frac{1}{i}\left( \frac{\partial }{\partial t}e^{B}\right)
e^{-B}+e^{B}\left( H_{0}-\frac{1}{4!}\frac{1}{6}[A,[A,[A,[A,V]]]]\right)
e^{-B}\text{ }.  \label{Formula:DefinitonOfLQ}
\end{equation}

At this point, we proceed to seek a equation for $k$ s.t. the coefficient of 
$a_{x}^{\ast }a_{y}^{\ast }$ in $\frac{1}{4!}\frac{1}{6}%
e^{B}[A,[A,[A,[A,V]]]]e^{-B}$ is eliminated.

\begin{remark}
One might be concerned of the pure creation $a_{x}^{\ast }a_{y}^{\ast
}a_{z}^{\ast }$ in $[A,[A,[A,V]]].$ Lemma \ref{Lemma:KeyLemmaForErrorTerms}
and the factor $1/\sqrt{N}$ will take care of that. Note that \cite%
{GMM1,GMM2} do not have terms like this.
\end{remark}

\subsection{Equation for $k$}

\subsubsection{The infinitesimal metaplectic representation\protect\cite%
{GMM1}}

Let $sp$ be the infinite dimensional Lie algebra of matrices of the form%
\begin{equation*}
S(d,k,l)=\left( 
\begin{matrix}
d & k \\ 
l & -d^{T}%
\end{matrix}%
\right)
\end{equation*}%
where $k$ and $l$ are symmetric, and $Quad$ be the Lie algebra consisting of
homogeneous quadratics of the form 
\begin{align*}
Q(d,k,l):=& \frac{1}{2}\left( 
\begin{matrix}
a_{x} & a_{x}^{\ast }%
\end{matrix}%
\right) \left( 
\begin{matrix}
d & k \\ 
l & -d^{T}%
\end{matrix}%
\right) \left( 
\begin{matrix}
-a_{y}^{\ast } \\ 
a_{y}%
\end{matrix}%
\right) \\
& =-\int d(x,y)\frac{a_{x}a_{y}^{\ast }+a_{y}^{\ast }a_{x}}{2}dxdy+\frac{1}{2%
}\int k(x,y)a_{x}a_{y}dxdy \\
& -\frac{1}{2}\int l(x,y)a_{x}^{\ast }a_{y}^{\ast }dxdy
\end{align*}%
equipped with Poisson bracket. In the spirit of page 185, Folland \cite%
{Folland}, we define the infinitesimal metaplectic representation: a Lie
algebra isomorphism ${\mathcal{I}}:sp\rightarrow Quad$ by $Q(d,k,l)={%
\mathcal{I}}(S(d,k,l)).$ Then we see that%
\begin{equation*}
B={\mathcal{I}}(K)\text{ },
\end{equation*}%
for 
\begin{equation}
K=\left( 
\begin{matrix}
0 & k(t,x,y) \\ 
\overline{k}(t,x,y) & 0%
\end{matrix}%
\right) \text{,}  \label{Formula:DefinitionOfK}
\end{equation}%
and it follows that

(i) 
\begin{equation*}
{\mathcal{I}}\left( e^{S}Ce^{-S}\right) =e^{{\mathcal{I(S)}}}{\mathcal{I}}%
\left( C\right) e^{-{\mathcal{I(S)}}}
\end{equation*}%
if ${\mathcal{I}}\left( C\right) \in sp.$

(ii) 
\begin{equation*}
{\mathcal{I}}\left( \left( \frac{\partial }{\partial t}e^{S}\right)
e^{-S}\right) =\left( \frac{\partial }{\partial t}e^{{\mathcal{I(S)}}%
}\right) e^{-{\mathcal{I(S)}}}
\end{equation*}%
if ${\mathcal{I(S)}}$ is skew-Hermitian.

(iii) 
\begin{equation*}
e^{{\mathcal{I(S)}}}%
\begin{pmatrix}
a_{x} & a_{x}^{\ast }%
\end{pmatrix}%
\begin{pmatrix}
f \\ 
g%
\end{pmatrix}%
e^{-{\mathcal{I(S)}}}=%
\begin{pmatrix}
a_{x} & a_{x}^{\ast }%
\end{pmatrix}%
e^{S}%
\begin{pmatrix}
f \\ 
g%
\end{pmatrix}%
\end{equation*}%
if ${\mathcal{I(S)}}$ is skew-Hermitian.

\begin{remark}
Properties (i) and (ii) will be used below. (iii) will be used in Section %
\ref{section:ErrorEstimates}.
\end{remark}

\subsubsection{Derivation of Equation \protect\ref{Equation:EquationOfk}}

Use the simplifications noted in Remark \ref{Remark:Simplicity}, recall that 
\begin{eqnarray*}
&&\frac{1}{4!6}[A,[A,[A,[A,V]]]] \\
&=&\frac{1}{2}\int v_{3,1-2,1-3}\left\vert \phi _{1}\right\vert ^{2}(%
\overline{\phi }_{2}\overline{\phi }_{3}a_{y}a_{z}+\phi _{2}\phi
_{3}a_{y}^{\ast }a_{z}^{\ast })dxdydz \\
&&+\int v_{3,1-2,1-3}\left\vert \phi _{1}\right\vert ^{2}\overline{\phi }%
_{2}\phi _{3}a_{z}^{\ast }a_{y}dxdydz \\
&&+\frac{1}{2}\int v_{3,1-2,1-3}\left\vert \phi _{1}\right\vert
^{2}\left\vert \phi _{2}\right\vert ^{2}a_{z}^{\ast }a_{z}dxdydz\text{ ,}
\end{eqnarray*}%
and 
\begin{equation*}
H_{0}=\int a_{x}^{\ast }\Delta a_{x}dx
\end{equation*}%
we write%
\begin{equation*}
G=\left( 
\begin{matrix}
g & 0 \\ 
0 & -g^{T}%
\end{matrix}%
\right) \qquad \text{and}\qquad M=\left( 
\begin{matrix}
0 & m \\ 
-\overline{m} & 0%
\end{matrix}%
\right)
\end{equation*}%
with 
\begin{equation*}
g=-\triangle \delta _{1-2}+\left( \int v_{3,1-2,1-3}\left\vert \phi
_{3}\right\vert ^{2}dz\right) \overline{\phi }_{1}\phi _{2}+\frac{1}{2}%
\left( \int v_{3,1-2,1-3}\left\vert \phi _{2}\right\vert ^{2}\left\vert \phi
_{3}\right\vert ^{2}dydz\right) \delta _{1-2}
\end{equation*}%
and%
\begin{equation*}
m=-\left( \int v_{3,1-2,1-3}\left\vert \phi _{3}\right\vert ^{2}dz\right) 
\overline{\phi }_{1}\overline{\phi }_{2}\text{.}
\end{equation*}

Of course, we would like to be able to write

\begin{equation*}
H_{0}-\frac{1}{4!}\frac{1}{6}[A,[A,[A,[A,V]]]]={\mathcal{I}}\left( G\right) +%
{\mathcal{I}}\left( M\right) \text{ .}
\end{equation*}%
Unfortunately, the above equality is not true. For example%
\begin{equation*}
{\mathcal{I}}\left( \left( 
\begin{matrix}
-\triangle \delta _{1-2} & 0 \\ 
0 & \triangle \delta _{1-2}%
\end{matrix}%
\right) \right) =\int \frac{a_{x}^{\ast }\triangle a_{x}+a_{x}\triangle
a_{x}^{\ast }}{2}dx.
\end{equation*}%
However, the commutators of ${\mathcal{I}}\left( G\right) ,$ ${\mathcal{I}}%
\left( M\right) $ and $H_{0}-\frac{1}{4!}\frac{1}{6}[A,[A,[A,[A,V]]]]$ with $%
B$ are the same as in the discussion in page 287 in \cite{GMM1}. The same
idea applies here.

Split 
\begin{equation*}
H_{0}-\frac{1}{4!}\frac{1}{6}[A,[A,[A,[A,V]]]]=H_{G}+{\mathcal{I}}(\mathcal{M%
})
\end{equation*}%
where 
\begin{equation*}
H_{G}=H_{0}-\int v_{3,1-2,1-3}\left\vert \phi _{3}\right\vert ^{2}\overline{%
\phi }_{2}\phi _{1}a_{x}^{\ast }a_{y}dxdydz-\frac{1}{2}\int
v_{3,1-2,1-3}\left\vert \phi _{2}\right\vert ^{2}\left\vert \phi
_{3}\right\vert ^{2}a_{x}^{\ast }a_{x}dxdydz
\end{equation*}%
which has the property that%
\begin{equation*}
\lbrack H_{G},B]=[\mathcal{I}\left( G\right) ,B].
\end{equation*}

Now, $L_{Q}$ from formula \ref{Formula:DefinitonOfLQ} reads 
\begin{eqnarray*}
L_{Q} &=&\frac{1}{i}\left( \frac{\partial }{\partial t}e^{B}\right)
e^{-B}+e^{B}\left( H_{0}-\frac{1}{4!}\frac{1}{6}[A,[A,[A,[A,V]]]]\right)
e^{-B} \\
&=&{\mathcal{I}}\left( \left( \frac{1}{i}\frac{\partial }{\partial t}%
e^{K}\right) e^{-K}\right) +e^{B}H_{G}e^{-B}+{\mathcal{I}}\left(
e^{K}Me^{-K}\right) \\
&=&{\mathcal{I}}\left( \left( \frac{1}{i}\frac{\partial }{\partial t}%
e^{K}\right) e^{-K}\right) +H_{G}+\left( e^{B}H_{G}e^{-B}-H_{G}\right) +{%
\mathcal{I}}\left( e^{K}Me^{-K}\right) \\
&=&H_{G}+{\mathcal{I}}\left( \left( \frac{1}{i}\frac{\partial }{\partial t}%
e^{K}\right) e^{-K}\right) +[e^{B},H_{G}]e^{-B}+{\mathcal{I}}\left(
e^{K}Me^{-K}\right) \\
&=&H_{G}+{\mathcal{I}}\left( \left( \frac{1}{i}\frac{\partial }{\partial t}%
e^{K}\right) e^{-K}+[e^{K},G]e^{-K}+e^{K}Me^{-K}\right) \\
&=&H_{G}+{\mathcal{I}}(\mathcal{M}_{1}+\mathcal{M}_{2}+\mathcal{M}_{3})\text{
}.
\end{eqnarray*}

Then by the definition of the isomorphism ${\mathcal{I}}$, the coefficient
of $a_{x}a_{y}$ is $-(\mathcal{M}_{1}+\mathcal{M}_{2}+\mathcal{M}_{3})_{12},$
and the coefficient of $a_{x}^{\ast }a_{y}^{\ast }$ is $(\mathcal{M}_{1}+%
\mathcal{M}_{2}+\mathcal{M}_{3})_{21}.$ To write it explicitly:%
\begin{align}
& -(\mathcal{M}_{1}+\mathcal{M}_{2}+\mathcal{M}_{3})_{12}
\label{Formula:CoeffOfOnlyCreation} \\
& =\overline{(\mathcal{M}_{1}+\mathcal{M}_{2}+\mathcal{M}_{3})_{21}}  \notag
\\
& =(i\sinh (k)_{t}+\sinh (k)g^{T}+g\sinh (k))\overline{\cosh (k)}  \notag \\
& -(i\cosh (k)_{t}-[\cosh (k),g]\sinh (k)  \notag \\
& -\sinh (k)\overline{m}\sinh (k)-\cosh (k)m\overline{\cosh (k)}~.  \notag
\end{align}%
Setting formula \ref{Formula:CoeffOfOnlyCreation} to $0$ confers equation %
\ref{Equation:EquationOfk}. This implies that%
\begin{eqnarray*}
{\mathcal{I}}(\mathcal{M}_{1}+\mathcal{M}_{2}+\mathcal{M}_{3}) &=&-\int (%
\mathcal{M}_{1}+\mathcal{M}_{2}+\mathcal{M}_{3})_{11}\frac{a_{x}a_{y}^{\ast
}+a_{y}^{\ast }a_{x}}{2}dxdy \\
&=&-\int d(t,x,y)\frac{a_{x}a_{y}^{\ast }+a_{y}^{\ast }a_{x}}{2}dxdy \\
&=&-\int d(t,x,y)a_{y}^{\ast }a_{x}dxdy-\frac{1}{2}\int d(t,x,x)dx
\end{eqnarray*}%
where $d(t,x,y)$ is given by formula \ref{Formula:DefinitionOFd}.

\begin{remark}
$(\mathcal{M}_{1}+\mathcal{M}_{2}+\mathcal{M}_{3})_{ij}$ means the entry on
the $i$th row and the $j$th column of the matrix $(\mathcal{M}_{1}+\mathcal{M%
}_{2}+\mathcal{M}_{3}).$
\end{remark}

We summarize the computations we have done so far in this proposition:

\begin{proposition}
If $\phi $ and $k$ solve equations \ref{Equation:QuinticHartreeEquation} and %
\ref{Equation:EquationOfk}, then the coefficients of $a_{x}a_{y}$ and $%
a_{x}^{\ast }a_{y}^{\ast }$ in $e^{B}[A,[A,[A,[A,V]]]]e^{-B}$ are $0$ and $%
L_{Q}$ becomes 
\begin{align*}
L_{Q}=& H_{0}-\int v_{3}(x-y,y-z)\left\vert \phi (z)\right\vert ^{2}%
\overline{\phi }(y)\phi (x)a_{x}^{\ast }a_{y}dxdydz \\
& -\frac{1}{2}\int v_{3}(x-y,y-z)\left\vert \phi (y)\right\vert
^{2}\left\vert \phi (z)\right\vert ^{2}a_{x}^{\ast }a_{x}dxdydz \\
& -\int d(t,x,y)a_{y}^{\ast }a_{x}dxdy-\frac{1}{2}\int d(t,x,x)dx\text{ }.
\end{align*}%
Recall that $\Psi =e^{B}\Psi _{0}=e^{B}e^{\sqrt{N}A(t)}e^{itH_{N}}e^{-\sqrt{N%
}A(0)}e^{-B(0)}\Omega $ solves%
\begin{equation}
\frac{1}{i}\frac{\partial }{\partial t}\Psi =L\Psi .  \notag
\end{equation}%
We can now write out%
\begin{eqnarray*}
L &=&L_{Q}-\frac{1}{6}N^{-2}e^{B}Ve^{-B}-\frac{1}{6}N^{-3/2}e^{B}[A,V]e^{-B}-%
\frac{1}{12}N^{-1}e^{B}[A,[A,V]]e^{-B} \\
&&-\frac{1}{36}N^{-\frac{1}{2}}e^{B}[A,[A,[A,V]]]e^{-B}-N\chi _{0} \\
&=&H_{0}-\int v_{3}(x-y,y-z)\left\vert \phi (z)\right\vert ^{2}\overline{%
\phi }(y)\phi (x)a_{x}^{\ast }a_{y}dxdydz \\
&&-\frac{1}{2}\int v_{3}(x-y,y-z)\left\vert \phi (y)\right\vert
^{2}\left\vert \phi (z)\right\vert ^{2}a_{x}^{\ast }a_{x}dxdydz \\
&&-\int d(t,x,y)a_{y}^{\ast }a_{x}dxdy-\frac{1}{6}N^{-2}e^{B}Ve^{-B}-\frac{1%
}{6}N^{-3/2}e^{B}[A,V]e^{-B} \\
&&-\frac{1}{12}N^{-1}e^{B}[A,[A,V]]e^{-B}-\frac{1}{36}N^{-\frac{1}{2}%
}e^{B}[A,[A,[A,V]]]e^{-B}-\frac{1}{2}\int d(t,x,x)dx-N\chi _{0} \\
&=&\tilde{L}-\chi _{1}-N\chi _{0}
\end{eqnarray*}%
if we write 
\begin{eqnarray*}
\chi _{0}(t) &=&-\frac{1}{3}\int v_{3}(x-y,x-z)\left\vert \phi
(x)\right\vert ^{2}\left\vert \phi (y)\right\vert ^{2}\left\vert \phi
(z)\right\vert ^{2}dxdydz \\
\chi _{1}(t) &=&\frac{1}{2}\int d(t,x,x)dx.
\end{eqnarray*}
\end{proposition}

\begin{remark}
Note that $\left( \tilde{L}\right) ^{\ast }=\tilde{L}$ and $\tilde{L}$
commutes with functions of time. This is needed in the proof of Theorem \ref%
{THM:2ndOrderCorrection} which is below.
\end{remark}

\subsubsection{The proof of Theorem \protect\ref{THM:2ndOrderCorrection}}

Applying the above proposition, we can give the proof of Theorem \ref%
{THM:2ndOrderCorrection} at this point.%
\begin{eqnarray*}
&&\Vert e^{-\sqrt{N}A(t)}e^{-B(t)}e^{-i\int_{0}^{t}(N\chi _{0}(s)+\chi
_{1}(s))ds}\Omega -e^{itH_{N}}e^{-\sqrt{N}A(0)}e^{-B(0)}\Omega \,\Vert _{{%
\mathcal{F}}} \\
&=&\Vert \Omega -e^{i\int_{0}^{t}(N\chi _{0}(s)+\chi _{1}(s))ds}e^{B(t)}e^{%
\sqrt{N}A(t)}e^{itH_{N}}e^{-\sqrt{N}A(0)}e^{-B(0)}\Omega \,\Vert _{{\mathcal{%
F}}} \\
&=&\Vert \Omega -e^{i\int_{0}^{t}(N\chi _{0}(s)+\chi _{1}(s))ds}\Psi \,\Vert
_{{\mathcal{F}}}
\end{eqnarray*}%
since $e^{-\sqrt{N}A(t)}$ and $e^{-B(t)}$ are unitary.

But%
\begin{eqnarray*}
&&\frac{\partial }{\partial t}\Vert \Omega -e^{i\int_{0}^{t}(N\chi
_{0}(s)+\chi _{1}(s))ds}\Psi \,\Vert _{{\mathcal{F}}}^{2} \\
&=&2\func{Re}\left( \frac{\partial }{\partial t}\left(
e^{i\int_{0}^{t}(N\chi _{0}(s)+\chi _{1}(s))ds}\Psi -\Omega \right)
,e^{i\int_{0}^{t}(N\chi _{0}(s)+\chi _{1}(s))ds}\Psi -\Omega \right) \\
&=&2\func{Re}\left( \left( \frac{\partial }{\partial t}-i\tilde{L}\right)
\left( e^{i\int_{0}^{t}(N\chi _{0}(s)+\chi _{1}(s))ds}\Psi -\Omega \right)
,e^{i\int_{0}^{t}(N\chi _{0}(s)+\chi _{1}(s))ds}\Psi -\Omega \right) \\
&=&2\func{Re}\left( i\tilde{L}\Omega ,e^{i\int_{0}^{t}(N\chi _{0}(s)+\chi
_{1}(s))ds}\Psi -\Omega \right) \\
&\leqslant &2\Vert \tilde{L}\Omega \,\Vert _{{\mathcal{F}}}\Vert
e^{i\int_{0}^{t}(N\chi _{0}(s)+\chi _{1}(s))ds}\Psi -\Omega \,\Vert _{{%
\mathcal{F}}}
\end{eqnarray*}%
due to the fact that%
\begin{equation*}
\left( \frac{1}{i}\frac{\partial }{\partial t}-\tilde{L}\right)
(e^{i\int_{0}^{t}(N\chi _{0}(s)+\chi _{1}(s))ds}\Psi -\Omega )=\tilde{L}%
\Omega .
\end{equation*}

Notice that 
\begin{eqnarray*}
\widetilde{L}\Omega &=&-\bigg(\frac{1}{6N^{2}}e^{B}Ve^{-B}+\frac{1}{6}%
N^{-3/2}e^{B}[A,V]e^{-B}+ \\
&&\frac{1}{12}N^{-1}e^{B}[A,[A,V]]e^{-B}+\frac{1}{36}N^{-\frac{1}{2}%
}e^{B}[A,[A,[A,V]]]e^{-B}\bigg)\Omega ,
\end{eqnarray*}%
we reach 
\begin{eqnarray*}
&&\frac{\partial }{\partial t}\Vert \Omega -e^{i\int_{0}^{t}(N\chi
_{0}(s)+\chi _{1}(s))ds}\Psi \,\Vert _{{\mathcal{F}}} \\
&\leqslant &\frac{\Vert e^{B}Ve^{-B}\Omega \Vert _{{\mathcal{F}}}}{6N^{2}}+%
\frac{\Vert e^{B}[A,V]e^{-B}\Omega \Vert _{{\mathcal{F}}}}{6N^{\frac{3}{2}}}
\\
&&+\frac{\Vert e^{B}[A,[A,V]]e^{-B}\Omega \Vert _{{\mathcal{F}}}}{12N}+\frac{%
\Vert e^{B}[A,[A,[A,V]]]e^{-B}\Omega \Vert _{{\mathcal{F}}}}{36N^{\frac{1}{2}%
}}
\end{eqnarray*}%
Whence we complete the proof of Theorem \ref{THM:2ndOrderCorrection} because 
$e^{-\sqrt{N}A(t)}e^{-B(t)}e^{-i\int_{0}^{t}(N\chi _{0}(s)+\chi
_{1}(s))ds}\Omega $ and $e^{itH_{N}}e^{-\sqrt{N}A(0)}e^{-B(0)}\Omega $ share
the same initial data $e^{-\sqrt{N}A(0)}e^{-B(0)}\Omega $.

\section{Solving Equation \protect\ref{Equation:EquationOfk} / Proof of
Theorem \protect\ref{THM:MainTheoremErrorEstimate} $\left( \text{Part I}%
\right) $}

\label{section:EstimatesFor u}

Starting from this section, we begin the proof of Theorem \ref%
{THM:MainTheoremErrorEstimate}. In other words, we are assuming that%
\begin{equation*}
v_{3}(x-y,x-z)=v(x-y,x-z)
\end{equation*}%
where $v$ is defined in formula \ref{Formula:DefinitonOfv-epsilon}.

We first study equation \ref{Equation:EquationOfk}. We prove an apriori
estimate for $u=\sinh (k)$ and use it in a Duhamel iteration argument to
show global existence. Finally we verify that $\int d(t,x,x)dx$ is locally
integrable in time.

Written in the notations in Remark \ref{Remark:Simplicity}, equation \ref%
{Equation:EquationOfk} reads%
\begin{equation*}
\left( iu_{t}+ug^{T}+gu-(I+p)m\right) =\left( ip_{t}+[g,p]+u\overline{m}%
\right) \left( I+p\right) ^{-1}u\text{,}
\end{equation*}%
where%
\begin{align*}
u(t,x,y)& =\sinh (k)=k+\frac{1}{3!}k\overline{k}k+\ldots ~, \\
\cosh (k)(t,x,y)& =I+p(t,x,y)=\delta _{1-2}+\frac{1}{2!}k\overline{k}+\ldots
~, \\
g(t,x,y)& =-\triangle \delta _{1-2}+\left( \int v_{1-2,1-3}\left\vert \phi
_{3}\right\vert ^{2}dz\right) \overline{\phi }_{1}\phi _{2} \\
& +\frac{1}{2}\left( \int v_{1-2,1-3}\left\vert \phi _{2}\right\vert
^{2}\left\vert \phi _{3}\right\vert ^{2}dydz\right) \delta _{1-2} \\
m(t,x,y)& =-\left( \int v_{1-2,1-3}\left\vert \phi _{3}\right\vert
^{2}dz\right) \overline{\phi }_{1}\overline{\phi }_{2}.
\end{align*}%
As mentioned in Theorem \ref{THM:2ndOrderCorrection}, we write composition
of kernels as products in the above e.g.%
\begin{equation*}
k\overline{k}(x,y)=\int k(x,z)\overline{k}(z,y)dz.
\end{equation*}%
Observe that $g(t,x,y)=\overline{g(t,y,x)}$, i.e. $g^{\ast }=g$; and $%
m(t,x,y)=m(t,y,x)$, i.e. $m^{T}=m$. Moreover, $u^{T}=u$, $p^{\ast }=p$
because $k\in L_{s}^{2}(dxdy).$

Via $e^{K}e^{-K}=I$ with $K$ defined in formula \ref{Formula:DefinitionOfK},
we obtain the trigonometric identity%
\begin{eqnarray*}
u\overline{u} &=&\cosh (k)\overline{\cosh (k)}-I \\
&=&2p+p^{2}
\end{eqnarray*}%
which is a relation between $u$ and $p.$

\subsection{An Apriori Estimate of $u$}

\begin{theorem}
\label{THM:CorollaryOfPrioriEstimateOfu}Let $v_{3}(x-y,x-z)=v(x-y,x-z).$ If $%
u=\sinh (k)$ is a solution of equation \ref{Equation:EquationOfk} on some
time interval $[0,T]$, then there exists a $C\geqslant 0$, independent of $%
T, $ s.t.%
\begin{equation*}
\Vert u(T)\Vert _{L_{(x,y)}^{2}}\leqslant C\left( 1+\Vert u(0)\Vert
_{L_{(x,y)}^{2}}\right) .
\end{equation*}
\end{theorem}

The major observation is the following lemma which is also the cornerstone
to showing Theorem \ref{THM:ExistenceOfU}.

\begin{lemma}
\label{Lemma:KeyLemmaForPrioriEstimateOfU}\cite{GMM2}From equation \ref%
{Equation:EquationOfk}, we deduce%
\begin{equation}
(ip_{t}+[g,p]+u\overline{m})(I+p)^{-1}=-(I+p)^{-1}\left( ip_{t}+[g,p]-m\bar{u%
}\right)  \label{Formula:Eq1ForKeyLemmaOfPrioriEstimateOfU}
\end{equation}%
and consequently%
\begin{equation}
i(u\overline{u})_{t}+[g,u\overline{u}]=m\bar{u}(I+p)-(I+p)u\overline{m}.
\label{Formula:Eq2ForKeyLemmaOfPrioriEstimateOfU}
\end{equation}
\end{lemma}

\begin{proof}
Multiply equation \ref{Equation:EquationOfk} on the right by $\overline{u}$,
it reads 
\begin{equation}
\left( iu_{t}+ug^{T}+gu\right) \overline{u}-(I+p)m\overline{u}%
=(ip_{t}+[g,p]+u\overline{m})(I+p)^{-1}u\overline{u}.
\label{equation:TimesUbarInPrioriEstimateOfU}
\end{equation}

Take the adjoint in the operator kernel sense of equation \ref%
{Equation:EquationOfk}, multiply on the left by $u,$ i.e. 
\begin{equation}
u\left( -i\bar{u}_{t}+g^{T}\bar{u}+\bar{u}g\right) -u\overline{m}(I+p)=u%
\overline{u}(I+p)^{-1}\left( -ip_{t}-[g,p]+m\overline{u}\right) .
\label{equation:TimesUInPrioriEstimateOfU}
\end{equation}%
Subtracting equations \ref{equation:TimesUbarInPrioriEstimateOfU} and \ref%
{equation:TimesUInPrioriEstimateOfU}, we have%
\begin{eqnarray}
&&i(u\overline{u})_{t}+[g,u\overline{u}]-(I+p)m\overline{u}+u\overline{m}%
(I+p)  \label{equation:MiddleOfULemma} \\
&=&(ip_{t}+[g,p]+u\overline{m})(I+p)^{-1}u\overline{u}-u\overline{u}%
(I+p)^{-1}\left( -ip_{t}-[g,p]+m\overline{u}\right)  \notag
\end{eqnarray}%
With $u\overline{u}=\cosh (k)\overline{\cosh (k)}-I$ and $u\overline{u}%
=2p+p^{2},$ we compute%
\begin{equation*}
(I+p)^{-1}u\overline{u}-(I+p)=(I+p)^{-1}=u\overline{u}(I+p)^{-1}-(I+p)
\end{equation*}%
and 
\begin{equation*}
(I+p)^{-1}u\overline{u}=(I+p)^{-1}p+p=u\overline{u}(I+p)^{-1}
\end{equation*}%
which transform equation \ref{equation:MiddleOfULemma} to%
\begin{eqnarray*}
&&i(2p+p^{2})_{t}+[g,2p+p^{2}]+u\overline{m}(I+p)^{-1}-(I+p)^{-1}m\overline{u%
} \\
&=&(ip_{t}+[g,p])((I+p)^{-1}p+p)-((I+p)^{-1}p+p)\left( -ip_{t}-[g,p]\right)
\end{eqnarray*}%
i.e.%
\begin{eqnarray*}
&&2(ip_{t}+[g,p])+u\overline{m}(I+p)^{-1}-(I+p)^{-1}m\overline{u} \\
&=&(ip_{t}+[g,p])(I+p)^{-1}p+(I+p)^{-1}p\left( ip_{t}+[g,p]\right)
\end{eqnarray*}%
which is equation \ref{Formula:Eq1ForKeyLemmaOfPrioriEstimateOfU} due to $%
I-(I+p)^{-1}p=(I+p)^{-1}.$

Multiplying equation \ref{Formula:Eq1ForKeyLemmaOfPrioriEstimateOfU} on the
right and left by $(I+p)$ produces%
\begin{equation*}
(I+p)(ip_{t}+[g,p]+u\overline{m})=-\left( ip_{t}+[g,p]-m\bar{u}\right) (I+p)
\end{equation*}%
i.e. equation \ref{Formula:Eq2ForKeyLemmaOfPrioriEstimateOfU}: 
\begin{equation*}
i(u\overline{u})_{t}+[g,u\overline{u}]=m\bar{u}(I+p)-(I+p)u\overline{m}
\end{equation*}%
because $u\overline{u}=2p+p^{2}$.
\end{proof}

Taking the trace in formula \ref{Formula:Eq2ForKeyLemmaOfPrioriEstimateOfU}
yields 
\begin{equation*}
{\frac{d}{dt}}\Vert u\Vert _{L^{2}}^{2}=Tr\left[ (1/i)\big(m\overline{u}%
(1+p)-(1+p)u\overline{m}\big)\right] \ .
\end{equation*}%
Note that 
\begin{eqnarray*}
\Vert u\Vert _{L^{2}}^{2} &=&Tr\left( u\overline{u}\right) \\
&=&2Tr\left( p\right) +Tr\left( p^{2}\right) \\
&\geqslant &\Vert p\Vert _{L^{2}}^{2}
\end{eqnarray*}%
because $p(t,x,y)=\frac{1}{2!}k\overline{k}+\ldots $ must have a nonnegative
trace. So%
\begin{eqnarray*}
{\frac{d}{dt}}\Vert u\Vert _{L^{2}}^{2} &\leqslant &2\left( \Vert m\Vert
_{L^{2}}\Vert u\Vert _{L^{2}}+\Vert m\Vert _{L^{2}}\Vert u\Vert
_{L^{2}}\Vert p\Vert _{L^{2}}\right) \\
&\leqslant &2\left( \Vert m\Vert _{L^{2}}\Vert u\Vert _{L^{2}}+\Vert m\Vert
_{L^{2}}\Vert u\Vert _{L^{2}}^{2}\right) .
\end{eqnarray*}%
By a Gronwall's inequality, we deduce 
\begin{equation*}
\Vert u(T)\Vert _{L_{(x,y)}^{2}}\leqslant \left( \int_{0}^{T}\Vert m\Vert
_{L_{(x,y)}^{2}}dt+\Vert u(0)\Vert _{L_{(x,y)}^{2}}\right) \exp \left(
\int\limits_{0}^{T}\Vert m\Vert _{L_{(x,y)}^{2}}dt\right) .
\end{equation*}%
The following lemma gives us Theorem \ref{THM:CorollaryOfPrioriEstimateOfu}.

\begin{lemma}
\label{Lemma:EstimateForM}If $v_{3}(x-y,x-z)=v(x-y,x-z),$ then 
\begin{equation*}
\Vert m\Vert _{L_{t}^{1}(\mathbb{R}^{+})L_{(x,y)}^{2}}\leqslant C<\infty
\end{equation*}
\end{lemma}

\begin{proof}
Because%
\begin{eqnarray*}
&&v(x-y,x-z) \\
&=&v_{0}(x-y)v_{0}(x-z)+v_{0}(x-y)v_{0}(y-z)+v_{0}(x-z)v_{0}(y-z),
\end{eqnarray*}%
we have 
\begin{eqnarray*}
\Vert m\Vert _{L_{(x,y)}^{2}}^{2} &=&\int \left( \int v(x-y,x-z)\left\vert
\phi _{3}\right\vert ^{2}dz\right) ^{2}|\phi _{1}|^{2}|\phi _{2}|^{2}dxdy \\
&\leqslant &C\int |\phi _{1}|^{2}|\phi _{2}|^{2}v_{0}^{2}(x-y)\left( \int
v_{0}(x-z)\left\vert \phi _{3}\right\vert ^{2}dz\right) ^{2}dxdy \\
&&+C\int |\phi _{1}|^{2}|\phi _{2}|^{2}v_{0}^{2}(x-y)\left( \int
v_{0}(y-z)\left\vert \phi _{3}\right\vert ^{2}dz\right) ^{2}dxdy \\
&&+C\int |\phi _{1}|^{2}|\phi _{2}|^{2}\left( \int
v_{0}(x-z)v_{0}(y-z)\left\vert \phi _{3}\right\vert ^{2}dz\right) ^{2}dxdy \\
&=&I+II+III.
\end{eqnarray*}%
A combination of H\"{o}lder and interpolation gives the following estimates%
\begin{eqnarray*}
I+II &=&2C\int |\phi _{1}|^{2}\left( \int v_{0}^{2}(x-y)|\phi
_{2}|^{2}dy\right) \left( \int v_{0}(x-z)\left\vert \phi _{3}\right\vert
^{2}dz\right) ^{2}dx \\
&\leqslant &C\left\Vert \phi \right\Vert _{L^{6}}^{2}\left\Vert \int
v_{0}^{2}(\cdot -y)|\phi _{2}|^{2}dy\right\Vert _{L^{\infty }}\left\Vert
\int v_{0}(\cdot -z)\left\vert \phi _{3}\right\vert ^{2}dz\right\Vert
_{L^{3}}^{2} \\
&\leqslant &C\left\Vert \phi \right\Vert _{L^{6}}^{2}\left\Vert \phi
_{0}\right\Vert _{L^{2}}^{2}\left\Vert \phi \right\Vert
_{L^{6}}^{4}\leqslant C\left\Vert \phi _{0}\right\Vert
_{L^{2}}^{2}\left\Vert \phi \right\Vert _{L^{6}}^{6}, \\
III &=&C\int v_{0}(x-z_{1})v_{0}(y-z_{1})v_{0}(x-z_{2})v_{0}(y-z_{2}) \\
&&|\phi _{1}|^{2}|\phi _{2}|^{2}\left\vert \phi (z_{1})\right\vert
^{2}\left\vert \phi (z_{2})\right\vert ^{2}dxdydz_{1}dz_{2} \\
&\leqslant &C\int dz_{1}dz_{2}\left\vert \phi (z_{1})\right\vert
^{2}\left\vert \phi (z_{2})\right\vert ^{2}\left( \int
v_{0}^{2}(x-z_{1})v_{0}^{2}(y-z_{1})|\phi _{1}|^{2}|\phi
_{2}|^{2}dxdy\right) ^{\frac{1}{2}} \\
&&\left( \int v_{0}^{2}(x-z_{2})v_{0}^{2}(y-z_{2})|\phi _{1}|^{2}|\phi
_{2}|^{2}dxdy\right) ^{\frac{1}{2}} \\
&=&C\left( \int dz\left\vert \phi (z)\right\vert ^{2}\left( \int
v_{0}^{2}(x-z)|\phi _{1}|^{2}dx\right) \right) ^{2} \\
&\leqslant &C\left\Vert |\phi |^{2}\right\Vert _{L^{3}}^{2}\left\Vert \int
v_{0}^{2}(x-z)|\phi _{1}|^{2}dx\right\Vert _{L^{\frac{3}{2}}}^{2} \\
&\leqslant &C\left\Vert \phi \right\Vert _{L^{6}}^{4}\left\Vert \phi
\right\Vert _{L^{3}}^{4}\leqslant C\left\Vert \phi _{0}\right\Vert
_{L^{2}}^{2}\left\Vert \phi \right\Vert _{L^{6}}^{6}.
\end{eqnarray*}%
i.e. $\Vert m\Vert _{L_{(x,y)}^{2}}\leqslant C\left\Vert \phi \right\Vert
_{L^{6}}^{3}\leqslant Ct^{-3}$, for $t\geqslant 1,$ by Theorem \ref%
{THM:TheLongTimeBehaviorOfHartree}. So we conclude the lemma.
\end{proof}

\begin{remark}
Theorem \ref{THM:CorollaryOfPrioriEstimateOfu} also has consequences on $p$
because $\Vert p\Vert _{L^{2}}\leqslant \Vert u\Vert _{L^{2}}.$
\end{remark}

\subsection{The Existence of $u$}

Because equation \ref{Equation:EquationOfk}%
\begin{equation*}
\left( iu_{t}+ug^{T}+gu-(I+p)m\right) =\left( ip_{t}+[g,p]+u\overline{m}%
\right) \left( I+p\right) ^{-1}u\text{,}
\end{equation*}%
is fully nonlinear in $k$, it is not easy to solve for $k$ directly from the
equation. However, if we put in%
\begin{equation*}
I+p=\cosh (k)=\sqrt{I+u\overline{u}}
\end{equation*}%
in the operator sense, equation \ref{Equation:EquationOfk} becomes a
quasilinear NLS equation in $u=\sinh (k)$. In fact, written out explicitly,
the left hand side of equation \ref{Equation:EquationOfk} is 
\begin{eqnarray}
iu_{t}+ug^{T}+gu &=&\left( i\frac{\partial }{\partial t}-\Delta _{x}-\Delta
_{y}\right) u(t,x,y)  \label{Formula:ExplicitGMMEq} \\
&&+\overline{\phi }_{1}\int \left( \int v(x-y_{1},x-z)\left\vert \phi
_{3}\right\vert ^{2}dz\right) \phi (y_{1})u(t,y_{1},y)dy_{1}  \notag \\
&&+\overline{\phi }_{2}\int u(t,x,y_{1})\left( \int v(y_{1}-y,x-z)\left\vert
\phi _{3}\right\vert ^{2}dz\right) \phi (y_{1})dy_{1}  \notag \\
&&+\frac{1}{2}\left( \int v_{\cdot -2,\cdot -3}\left\vert \phi
_{2}\right\vert ^{2}\left\vert \phi _{3}\right\vert ^{2}dydz\right)
(x)u(t,x,y)  \notag \\
&&+\frac{1}{2}\left( \int v_{\cdot -2,\cdot -3}\left\vert \phi
_{2}\right\vert ^{2}\left\vert \phi _{3}\right\vert ^{2}dydz\right)
(y)u(t,x,y)  \notag
\end{eqnarray}%
and the main term of the right hand side%
\begin{eqnarray}
ip_{t}+[g,p] &=&\left( i\frac{\partial }{\partial t}-\Delta _{x}+\Delta
_{y}\right) p(t,x,y)  \label{Formula:ExplicitGMMEqForP} \\
&&+\overline{\phi }_{1}\int \left( \int v(x-y_{1},x-z)\left\vert \phi
_{3}\right\vert ^{2}dz\right) \phi (y_{1})p(t,y_{1},y)dy_{1}  \notag \\
&&-\overline{\phi }_{2}\int p(t,x,y_{1})\left( \int v(y_{1}-y,x-z)\left\vert
\phi _{3}\right\vert ^{2}dz\right) \phi (y_{1})dy_{1}  \notag \\
&&+\frac{1}{2}\left( \int v_{\cdot -2,\cdot -3}\left\vert \phi
_{2}\right\vert ^{2}\left\vert \phi _{3}\right\vert ^{2}dydz\right)
(x)p(t,x,y)  \notag \\
&&-\frac{1}{2}\left( \int v_{\cdot -2,\cdot -3}\left\vert \phi
_{2}\right\vert ^{2}\left\vert \phi _{3}\right\vert ^{2}dydz\right)
(y)p(t,x,y).  \notag
\end{eqnarray}

For our purpose, obtaining some reasonable estimates of $u$ and $p=\cosh
(k)-I$ is enough. So we would like to get around solving for $k$ and go to $%
u $ directly.

But at first, we ask the following: $k$ certainly determines $u$, but does $%
u $ determine $k$? The proof of Theorem \ref{THM:2ndOrderCorrection}
actually needs a well-defined $k.$

We answer the above question by the following lemma:

\begin{lemma}
\cite{GMM2}The map 
\begin{equation*}
k\mapsto u=\sinh (k)
\end{equation*}%
is one to one, onto, continuous, with a continuous inverse, from symmetric
Hilbert-Schmidt kernels $k$ onto symmetric Hilbert-Schmidt kernels $u$.
\end{lemma}

\begin{proof}
The proof of this lemma is in \cite{GMM2}.
\end{proof}

Now we consider the existence of $u$ satisfying equation \ref%
{Equation:EquationOfk}. As asserted, equation \ref{Equation:EquationOfk} is
a quasilinear NLS of $u.$ However, we can transform it into a semilinear
equation which is easier to deal with, through the following lemma.

\begin{lemma}
\label{Lemma:TransformQuasiToSemi}\cite{GMM2}The following equations are
equivalent for a symmetric, Hilbert-Schmidt $u$:%
\begin{eqnarray}
iu_{t}+ug^{T}+gu &=&(I+p)m+\left( ip_{t}+[g,p]+u\overline{m}\right) \left(
I+p\right) ^{-1}u  \notag \\
iu_{t}+ug^{T}+gu &=&(I+p)m+\frac{1}{2}\left( \left( I+p\right) ^{-1}m%
\overline{u}+u\overline{m}\left( I+p\right) ^{-1}\right) u  \notag \\
&&+\frac{1}{2}\left[ ip_{t}+[g,p],\left( I+p\right) ^{-1}\right] u
\label{Formula:Eq1ForQuasiToSemiLemma} \\
iu_{t}+ug^{T}+gu &=&(I+p)m+\frac{1}{2}\left( \left( I+p\right) ^{-1}m%
\overline{u}+u\overline{m}\left( I+p\right) ^{-1}\right) u  \notag \\
&&+\frac{1}{2}\left[ W,\left( I+p\right) ^{-1}\right] u
\label{Formula:Eq2ForQuasiToSemiLemma}
\end{eqnarray}%
if we set%
\begin{eqnarray*}
W &:&={\frac{1}{2\pi i}}\int\limits_{\Gamma }\left( u\overline{u}-z\right)
^{-1}F\left( u\overline{u}-z\right) ^{-1}\sqrt{I+z}dz \\
F &:&=m\overline{u}(I+p)-(I+p)u\overline{m}
\end{eqnarray*}%
Here, $\Gamma $ is a contour enclosing the spectrum of the non-negative
Hilbert-Schmidt operator $u\overline{u}$.
\end{lemma}

\begin{proof}
(Sketch) Equation \ref{Formula:Eq1ForQuasiToSemiLemma} is the same as
equation \ref{Equation:EquationOfk}, suitably re-written. The keystone of
the proof is%
\begin{equation*}
ip_{t}+[g,p]=W.
\end{equation*}%
But%
\begin{eqnarray*}
ip_{t}+[g,p] &=&i\left( \sqrt{I+u\overline{u}}\right) _{t}+[g,\sqrt{I+u%
\overline{u}}] \\
&=&{\frac{1}{2\pi i}}\int\limits_{\Gamma }\left( u\overline{u}-z\right)
^{-1}(i(u\overline{u})_{t}+[g,u\overline{u}])\left( u\overline{u}-z\right)
^{-1}\sqrt{I+z}dz
\end{eqnarray*}%
because%
\begin{eqnarray*}
\sqrt{I+u\overline{u}} &=&-\frac{1}{2\pi i}\int\limits_{\Gamma }\left( u%
\overline{u}-z\right) ^{-1}\sqrt{I+z}dz \\
i\left( \left( u\overline{u}-z\right) ^{-1}\right) _{t}+[g,\left( u\overline{%
u}-z\right) ^{-1}] &=&-\left( u\overline{u}-z\right) ^{-1}(i(u\overline{u}%
)_{t}+[g,u\overline{u}])\left( u\overline{u}-z\right) ^{-1}.
\end{eqnarray*}%
The result follows from equation \ref%
{Formula:Eq2ForKeyLemmaOfPrioriEstimateOfU}%
\begin{equation*}
i(u\overline{u})_{t}+[g,u\overline{u}]=F=m\overline{u}(I+p)-(I+p)u\overline{m%
}.
\end{equation*}
\end{proof}

Whence, we only need to show the existence for equation \ref%
{Formula:Eq2ForQuasiToSemiLemma} which is of the form%
\begin{equation*}
iu_{t}+ug^{T}+gu=m+N(u)
\end{equation*}%
where the nonlinear part $N(u)$ involves no derivatives of $u.$ Via the
ordinary iteration procedure, we conclude the following existence theorem:

\begin{theorem}
\label{THM:ExistenceOfU}\cite{GMM2}Given $u_{0}\in L_{(x,y)}^{2}(\mathbb{R}%
^{6}\mathbb{)}$ symmetric, there exists $\varepsilon _{0}$ such that if 
\begin{equation*}
\Vert m\Vert _{L_{t}^{1}([0,T])L_{(x,y)}^{2}}\leqslant \varepsilon _{0}
\end{equation*}%
then there exists $u\in L_{t}^{\infty }([0,T])L_{(x,y)}^{2}$ solving
equation \ref{Formula:Eq2ForQuasiToSemiLemma} and hence equation \ref%
{Equation:EquationOfk} with prescribed initial condition $%
u(0,x,y)=u_{0}(x,y)\in L_{(x,y)}^{2}(\mathbb{R}^{6}\mathbb{)}$.
\end{theorem}

Since we have shown $\Vert m\Vert _{L_{t}^{1}(\mathbb{R}^{+})L_{(x,y)}^{2}}<%
\infty $ in Lemma \ref{Lemma:EstimateForM}, we can divide $\mathbb{R}^{+}$
into countably many time intervals $[T_{n},T_{n+1}]$ such that $\Vert m\Vert
_{L_{t}^{1}([T_{n},T_{n+1}])L_{(x,y)}^{2}}\leqslant \varepsilon _{0}$. So
the above existence theorem in fact implies the global existence of $u$ and
thus $p.$

Via Theorem \ref{THM:CorollaryOfPrioriEstimateOfu}, we have%
\begin{equation*}
\Vert u\Vert _{L_{t}^{\infty }(\mathbb{R}^{+})L_{(x,y)}^{2}}\leqslant C,
\end{equation*}%
which implies%
\begin{equation*}
\Vert p\Vert _{L_{t}^{\infty }(\mathbb{R}^{+})L_{(x,y)}^{2}}\leqslant C.
\end{equation*}%
Moreover, the following estimates hold.

\begin{theorem}
\label{THM:MorePrioriEstimateOfu}Let $u\in L_{t}^{\infty }(\mathbb{R}%
^{+})L_{(x,y)}^{2}$ be the solution of equation \ref{Equation:EquationOfk}
subject to $u_{0}\in L_{(x,y)}^{2}(\mathbb{R}^{6}\mathbb{)}$ described in
Theorem \ref{THM:ExistenceOfU}. Then $u$ satisfies the following additional
properties:%
\begin{eqnarray}
\Vert \left( i\frac{\partial }{\partial t}-\Delta _{x}-\Delta _{y}\right)
u\Vert _{L_{t}^{1}(\mathbb{R}^{+})L_{(x,y)}^{2}} &\leqslant &C
\label{Formula:Eq1ForExistenceOfU} \\
\Vert \left( i\frac{\partial }{\partial t}-\Delta _{x}+\Delta _{y}\right)
p\Vert _{L_{t}^{1}(\mathbb{R}^{+})L_{(x,y)}^{2}} &\leqslant &C
\label{Formula:Eq2ForExistenceOfU}
\end{eqnarray}%
where $C$ only depends on $v,$ $C_{1}$, $C_{2}$ and $\left\Vert
u_{0}\right\Vert _{L_{(x,y)}^{2}}$. See Theorem \ref{THM:MainTheorem} for $%
C_{1}$ and $C_{2}.$
\end{theorem}

\begin{proof}
We will only show estimate \ref{Formula:Eq1ForExistenceOfU}. Estimate \ref%
{Formula:Eq2ForExistenceOfU} can be shown similarly from%
\begin{equation*}
ip_{t}+[g,p]=W.
\end{equation*}

The proof is separated into 2 parts.

On the one hand we show%
\begin{equation*}
\Vert iu_{t}+ug^{T}+gu\Vert _{L_{t}^{1}(\mathbb{R}^{+})L_{(x,y)}^{2}}%
\leqslant \Vert m\Vert _{L_{t}^{1}(\mathbb{R}^{+})L_{(x,y)}^{2}}+\Vert
N(u)\Vert _{L_{t}^{1}(\mathbb{R}^{+})L_{(x,y)}^{2}}\leqslant C_{\varepsilon
}.
\end{equation*}%
On the other hand we control the terms in $iu_{t}+ug^{T}+gu$ different from $%
\left( i\frac{\partial }{\partial t}-\Delta _{x}-\Delta _{y}\right) u$,
namely 
\begin{equation*}
\int \left( \int v(x-y_{1},x-z)\left\vert \phi _{3}\right\vert ^{2}dz\right) 
\overline{\phi }(x)\phi (y_{1})u(y_{1},y)dy_{1}
\end{equation*}%
and 
\begin{equation*}
\frac{1}{2}\left( \int v(x-y,x-z)\left\vert \phi _{2}\right\vert
^{2}\left\vert \phi _{3}\right\vert ^{2}dydz\right) u.
\end{equation*}%
One sees the above two terms from formula \ref{Formula:ExplicitGMMEq}.

Part I. Recall that%
\begin{equation*}
N(u)=pm+\frac{1}{2}\left( \left( I+p\right) ^{-1}m\overline{u}+u\overline{m}%
\left( I+p\right) ^{-1}\right) u+\frac{1}{2}\left[ W,\left( I+p\right) ^{-1}%
\right] u.
\end{equation*}%
We have proven 
\begin{equation*}
\Vert p\Vert _{L_{t}^{\infty }(\mathbb{R}^{+})L_{(x,y)}^{2}}\leqslant \Vert
u\Vert _{L_{t}^{\infty }(\mathbb{R}^{+})L_{(x,y)}^{2}}\leqslant
C_{\varepsilon }.
\end{equation*}%
Together with the fixed time estimate:%
\begin{equation}
\left\Vert kl\right\Vert _{H-S}\leqslant \left\Vert k\right\Vert
_{op}\left\Vert l\right\Vert _{H-S}
\label{Estimate:InsideTheMainUEstimateOperator}
\end{equation}%
these take care of most of the terms in $N(u)$ because $\left( I+p\right)
^{-1}$ and $\left( u\overline{u}-z\right) ^{-1}|_{z\in \Gamma }$ have
uniformly bounded operator norms. In inequality \ref%
{Estimate:InsideTheMainUEstimateOperator}, $\left\Vert \cdot \right\Vert
_{H-S}$ stands for the Hilbert-Schmidt norm and $\left\Vert \cdot
\right\Vert _{op}$ stands for the operator norm. We only need to account for 
$W.$ However, the fact that $\left\vert z\right\vert \leqslant C\Vert u\Vert
_{L_{(x,y)}^{2}}^{2}$on $\Gamma \ $implies 
\begin{equation*}
\Vert W\Vert _{L_{t}^{1}(\mathbb{R}^{+})L_{(x,y)}^{2}}\leqslant C\left(
1+\Vert u\Vert _{L_{t}^{\infty }(\mathbb{R}^{+})L_{(x,y)}^{2}}^{6}\right)
\Vert m\Vert _{L_{t}^{1}(\mathbb{R}^{+})L_{(x,y)}^{2}}\leqslant C.
\end{equation*}%
i.e. $\Vert N(u)\Vert _{L_{t}^{1}(\mathbb{R}^{+})L_{(x,y)}^{2}}\leqslant C.$

Part II.

Using H\"{o}lder, it is not difficult to see the estimate%
\begin{eqnarray*}
&&\left\Vert \int \left( \int v(x-y_{1},x-z)\left\vert \phi _{3}\right\vert
^{2}dz\right) \overline{\phi }(x)\phi (y_{1})u(y_{1},y)dy_{1}\right\Vert
_{L_{t}^{1}(\mathbb{R}^{+})L_{(x,y)}^{2}} \\
&=&\left\Vert \left\{ \int \left\vert \int \left( \int
v(x-y_{1},x-z)\left\vert \phi _{3}\right\vert ^{2}dz\right) \overline{\phi }%
(x)\phi (y_{1})u(y_{1},y)dy_{1}\right\vert ^{2}dxdy\right\} ^{\frac{1}{2}%
}\right\Vert _{L_{t}^{1}(\mathbb{R}^{+})} \\
&\leqslant &\left\Vert \left\{ \int \left\vert \overline{\phi }%
(x)\right\vert ^{2}\left( \int \left( \int v(x-y_{1},x-z)\left\vert \phi
_{3}\right\vert ^{2}dz\right) ^{2}\left\vert \phi (y_{1})\right\vert
^{2}dy_{1}\right) \left( \int \left\vert u(y_{1},y)\right\vert
^{2}dy_{1}\right) dxdy\right\} ^{\frac{1}{2}}\right\Vert _{L_{t}^{1}(\mathbb{%
R}^{+})} \\
&=&\left\Vert \left\{ \int \left( \int v(x-y_{1},x-z)\left\vert \phi
_{3}\right\vert ^{2}dz\right) ^{2}\left\vert \overline{\phi }(x)\right\vert
^{2}\left\vert \phi (y_{1})\right\vert ^{2}dxdy_{1}\right\} ^{\frac{1}{2}%
}\Vert u\Vert _{L_{(x,y)}^{2}}\right\Vert _{L_{t}^{1}(\mathbb{R}^{+})} \\
&\leqslant &\Vert m\Vert _{L_{t}^{1}(\mathbb{R}^{+})L_{(x,y)}^{2}}\Vert
u\Vert _{L_{t}^{\infty }(\mathbb{R}^{+})L_{(x,y)}^{2}} \\
&\leqslant &C.
\end{eqnarray*}%
It remains to show: 
\begin{equation}
\left\Vert \left( \int v(x-y,x-z)\left\vert \phi _{2}\right\vert
^{2}\left\vert \phi _{3}\right\vert ^{2}dydz\right) u\right\Vert _{L_{t}^{1}(%
\mathbb{R}^{+})L_{(x,y)}^{2}}\leqslant C.
\label{Estimate:CanBeUsedReplacingU=Fi}
\end{equation}

Write%
\begin{eqnarray*}
&&\left\Vert \left( \int v(x-y,x-z)\left\vert \phi _{2}\right\vert
^{2}\left\vert \phi _{3}\right\vert ^{2}dydz\right) u\right\Vert
_{L_{(x,y)}^{2}} \\
&=&\left( \int \left\vert u(t,x,y)\right\vert ^{2}\left( \int
v(x-y,x-z)\left\vert \phi _{2}\right\vert ^{2}\left\vert \phi
_{3}\right\vert ^{2}dydz\right) ^{2}dxdy\right) ^{\frac{1}{2}} \\
&\leqslant &C\left( \int \left\vert u(t,x,y)\right\vert ^{2}\left( \int
v_{0}(x-y)v_{0}(x-z)\left\vert \phi _{2}\right\vert ^{2}\left\vert \phi
_{3}\right\vert ^{2}dydz\right) ^{2}dxdy\right) ^{\frac{1}{2}} \\
&&+C\left( \int \left\vert u(t,x,y)\right\vert ^{2}\left( \int
v_{0}(x-y)v_{0}(y-z)\left\vert \phi _{2}\right\vert ^{2}\left\vert \phi
_{3}\right\vert ^{2}dydz\right) ^{2}dxdy\right) ^{\frac{1}{2}} \\
&&+C\left( \int \left\vert u(t,x,y)\right\vert ^{2}\left( \int
v_{0}(x-z)v_{0}(y-z)\left\vert \phi _{2}\right\vert ^{2}\left\vert \phi
_{3}\right\vert ^{2}dydz\right) ^{2}dxdy\right) ^{\frac{1}{2}} \\
&=&I+II+III.
\end{eqnarray*}%
According to the estimate 
\begin{equation*}
\left\vert \int v_{0}(x-y)\left\vert \phi (y)\right\vert ^{2}dy\right\vert
\leqslant C\left\Vert \phi \right\Vert _{L^{6}}^{2}\leqslant Ct^{-2}
\end{equation*}%
we acquire%
\begin{equation*}
I=C\left( \int \left\vert u(t,x,y)\right\vert ^{2}\left( \int
v_{0}(x-y)\left\vert \phi _{2}\right\vert ^{2}dy\right) ^{4}dxdy\right) ^{%
\frac{1}{2}}\leqslant Ct^{-4}\Vert u\Vert _{L_{t}^{\infty }(\mathbb{R}%
^{+})L_{(x,y)}^{2}}\text{ for }t\geqslant 1,
\end{equation*}%
\begin{eqnarray*}
II+III &=&2C\left( \int \left\vert u(t,x,y)\right\vert ^{2}\left( \int
v_{0}(x-y)v_{0}(y-z)\left\vert \phi _{2}\right\vert ^{2}\left\vert \phi
_{3}\right\vert ^{2}dydz\right) ^{2}dxdy\right) ^{\frac{1}{2}} \\
&\leqslant &2C\left( \int \left\vert u(t,x,y)\right\vert ^{2}\left( \int
v_{0}(x-y)\left\vert \phi _{2}\right\vert ^{2}dy\right) ^{2}C\left\Vert \phi
\right\Vert _{L^{6}}^{4}dxdy\right) ^{\frac{1}{2}} \\
&\leqslant &Ct^{-4}\Vert u\Vert _{L_{t}^{\infty }(\mathbb{R}%
^{+})L_{(x,y)}^{2}}\text{ for }t\geqslant 1.
\end{eqnarray*}%
i.e. estimate \ref{Estimate:CanBeUsedReplacingU=Fi}%
\begin{equation*}
\left\Vert \left( \int v(x-y,x-z)\left\vert \phi _{2}\right\vert
^{2}\left\vert \phi _{3}\right\vert ^{2}dydz\right) u\right\Vert _{L_{t}^{1}(%
\mathbb{R}^{+})L_{(x,y)}^{2}}\leqslant C.
\end{equation*}
\end{proof}

\subsection{The Trace $\protect\int d(t,x,x)\ dx$}

Recall that 
\begin{align*}
d(t,x,y)=& \left( i\sinh (k)_{t}+\sinh (k)g^{T}+g\sinh (k)\right) \overline{%
\sinh (k)} \\
-& \left( i\cosh (k)_{t}+[g,\cosh (k)]\right) \cosh (k) \\
-& \sinh (k)\overline{m}\cosh (k)-\cosh (k)m\overline{\sinh (k)}.
\end{align*}%
defined by Formula \ref{Formula:DefinitionOFd}. Rewrite it as%
\begin{eqnarray*}
d(t,x,y) &=&\left( iu_{t}+ug^{T}+gu\right) \overline{u}-\left(
ip_{t}+[g,p]\right) (I+p) \\
&&-u\overline{m}(I+p)-\left( I+p\right) m\overline{u}
\end{eqnarray*}%
because $I$ commutes with everything and $I_{t}=0$.

Notice that if $k_{1}(x,y)\in L_{(x,y)}^{2}$ and $k_{2}(x,y)\in
L_{(x,y)}^{2} $ then 
\begin{equation*}
\int |k_{1}k_{2}|(x,x)dx=\int |\int k_{1}(x,y)k_{2}(y,x)dy|dx\leqslant \Vert
k_{1}\Vert _{L_{(x,y)}^{2}}\Vert k_{2}\Vert _{L_{(x,y)}^{2}}.
\end{equation*}

At this point, we have already shown that $m,iu_{t}+ug^{T}+gu,\ ip_{t}+[g,p]$
and $u\overline{m}\in L_{t}^{1}(\mathbb{R}^{+})L_{(x,y)}^{2}$ and $u,$ $p\in
L_{t}^{\infty }(\mathbb{R}^{+})L_{(x,y)}^{2}$. So except $\left(
ip_{t}+[g,p]\right) I,$ all traces in Formula \ref{Formula:DefinitionOFd}
are well-defined and integrable on $\mathbb{R}^{+}.$

However,%
\begin{equation*}
ip_{t}+[g,p]=W,
\end{equation*}%
for%
\begin{eqnarray*}
W &=&{\frac{1}{2\pi i}}\int\limits_{\Gamma }\left( u\overline{u}-z\right)
^{-1}F\left( u\overline{u}-z\right) ^{-1}\sqrt{I+z}dz \\
F &=&m\overline{u}(I+p)-(I+p)u\overline{m}.
\end{eqnarray*}%
Inside the contour integral of $W$, since $\left( u\overline{u}-z\right)
^{-1}|_{z\in \Gamma }$ has uniformly bounded operator norm and $\left\vert 
\sqrt{I+z}\right\vert \leqslant C\left( 1+\left\Vert u\right\Vert
_{L_{t}^{\infty }(\mathbb{R}^{+})L_{(x,y)}^{2}}\right) $, we are in fact
dealing with%
\begin{equation*}
(Bounded)(H-S)(H-S)(Bounded)
\end{equation*}%
where $H-S$ stands for Hilbert-Schmidt. But $(Bounded)(H-S)$ is
Hilbert-Schmidt. So we are looking at $(H-S)(H-S)$ which has a trace
well-defined and locally integrable in time.

\section{Error Estimates / Proof of Theorem \protect\ref%
{THM:MainTheoremErrorEstimate} $\left( \text{Part II}\right) $}

\label{section:ErrorEstimates}

We finish the proof of Theorem \ref{THM:MainTheoremErrorEstimate} with the
proposition below whose proof consists of classical techniques.

\begin{proposition}
\label{Prop:ErrorEstimateConclusion}Let $\phi $ to be the solution of the
Hartree equation subject to (i), (ii), and (iii). Assume we have 
\begin{eqnarray*}
\Vert \left( i\frac{\partial }{\partial t}+\Delta _{x}\right) \phi \Vert
_{L_{t}^{1}(\mathbb{R}^{+})L_{x}^{2}} &\leqslant &C_{3}\text{ }, \\
\Vert \left( i\frac{\partial }{\partial t}-\Delta _{x}-\Delta _{y}\right)
u\Vert _{L_{t}^{1}(\mathbb{R}^{+})L_{(x,y)}^{2}} &\leqslant &C_{4}\text{ },
\\
\Vert \left( i\frac{\partial }{\partial t}-\Delta _{x}+\Delta _{y}\right)
p\Vert _{L_{t}^{1}(\mathbb{R}^{+})L_{(x,y)}^{2}} &\leqslant &C_{5}~,
\end{eqnarray*}%
then we have the error estimates: 
\begin{eqnarray*}
\int \Vert e^{B}Ve^{-B}\Omega \Vert _{{\mathcal{F}}}\,\,dt &\leqslant &C \\
\int \Vert e^{B}[A,V]e^{-B}\Omega \Vert _{{\mathcal{F}}}\,\,dt &\leqslant &C
\\
\int \Vert e^{B}[A,[A,V]]e^{-B}\Omega \Vert _{{\mathcal{F}}}\,\,dt
&\leqslant &C \\
\int \Vert e^{B}[A,[A,[A,V]]]e^{-B}\Omega \Vert _{{\mathcal{F}}}\,\,dt
&\leqslant &C
\end{eqnarray*}%
where $C$ only depends on $v,$ $\phi ,$ $C_{3}$, $C_{4},$ $C_{5},$ and $%
\left\Vert u_{0}\right\Vert _{L_{(x,y)}^{2}}$.
\end{proposition}

\begin{remark}
We can prove%
\begin{equation*}
\Vert \left( i\frac{\partial }{\partial t}+\Delta _{x}\right) \phi \Vert
_{L_{t}^{1}(\mathbb{R}^{+})L_{x}^{2}}\leqslant C.
\end{equation*}%
with the same method to show estimate \ref{Estimate:CanBeUsedReplacingU=Fi}.
\end{remark}

\begin{remark}
Theorem \ref{THM:MorePrioriEstimateOfu} shows that $C_{4}$, $C_{5}$ depends
only on $v,$ $C_{1},\ C_{2}$ and $\left\Vert u_{0}\right\Vert
_{L_{(x,y)}^{2}}$. So $C$ here is determined by $v,$ $C_{1},\ C_{2}$ and $%
\left\Vert u_{0}\right\Vert _{L_{(x,y)}^{2}}$.
\end{remark}

\begin{remark}
For Theorem \ref{THM:MainTheorem}, we take $k(0,x,y)=0$ i.e. $u_{0}=0.$
\end{remark}

Ideally, we would like to prove Proposition \ref%
{Prop:ErrorEstimateConclusion} in complete details. However,%
\begin{eqnarray*}
e^{B}a_{x_{0}}^{\ast }e^{-B} &=&e^{B}%
\begin{pmatrix}
a_{x} & a_{x}^{\ast }%
\end{pmatrix}%
\begin{pmatrix}
0 \\ 
1%
\end{pmatrix}%
e^{-B}=%
\begin{pmatrix}
a_{x} & a_{x}^{\ast }%
\end{pmatrix}%
e^{K}%
\begin{pmatrix}
0 \\ 
1%
\end{pmatrix}
\\
&=&\int \left( u(x_{1},x_{0})a_{x_{0}}+\overline{\cosh (k)}%
(x_{1},x_{0})a_{x_{0}}^{\ast }\right) dx_{1}, \\
e^{B}a_{x_{0}}e^{-B} &=&e^{B}%
\begin{pmatrix}
a_{x} & a_{x}^{\ast }%
\end{pmatrix}%
\begin{pmatrix}
1 \\ 
0%
\end{pmatrix}%
e^{-B}=%
\begin{pmatrix}
a_{x} & a_{x}^{\ast }%
\end{pmatrix}%
e^{K}%
\begin{pmatrix}
1 \\ 
0%
\end{pmatrix}
\\
&=&\int \left( \cosh (k)(x_{2},x_{0})a_{x_{2}}+\overline{u}%
(x_{2},x_{0})a_{x_{2}}^{\ast }\right) dx_{2},
\end{eqnarray*}%
and%
\begin{equation*}
\cosh (k)(x,y)=\delta (x-y)+p(x,y),
\end{equation*}%
their products generate a large number of terms. The fact that we will
always commute the annihilations to the right, e.g. $a_{x_{1}}^{\ast
}a_{y_{2}}a_{z_{2}}^{\ast }=\delta (y_{2}-z_{2})a_{x_{1}}^{\ast
}+a_{x_{1}}^{\ast }a_{z_{2}}^{\ast }a_{y_{2}},$ to avoid $k(x,x)$ or related
traces, produces even more terms. Hence it is impractical to list every
single term in $e^{B}Ve^{-B}\Omega $ etc., instead, we prove a key lemma and
do a typical estimate.

\begin{lemma}
\label{Lemma:KeyLemmaForErrorTerms}(Key Lemma) Let $x_{1},y_{1},y_{2}\in 
\mathbb{R}^{3}$, $x_{2}\in \mathbb{R}^{n_{1}},y_{3}\in \mathbb{R}^{n_{2}}$
with the possibility that $n_{1}$ or $n_{2}$ is zero. Assume $f$, $g$ satisfy%
\begin{eqnarray*}
\Vert \left( i\frac{\partial }{\partial t}\pm \Delta _{x_{1}}\pm \Delta
_{x_{2}}\right) f(t,x_{1},x_{2})\Vert _{L_{t}^{1}(\mathbb{R}^{+})L_{x}^{2}}
&\leqslant &C, \\
\Vert \left( i\frac{\partial }{\partial t}\pm \left( \Delta _{y_{1}}+\Delta
_{y_{2}}\right) \pm \Delta _{y_{3}}\right) g(t,y_{1},y_{2},y_{3})\Vert
_{L_{t}^{1}(\mathbb{R}^{+})L_{y}^{2}} &\leqslant &C.
\end{eqnarray*}%
Moreover suppose $f|_{t=0}$, $g|_{t=0}\in L^{2}.$

Then%
\begin{equation*}
\int dt\left( \int v^{2}(x_{1}-y_{1},x_{1}-y_{2})\left\vert
f(t,x_{1},x_{2})\right\vert ^{2}\left\vert g(t,y_{1},y_{2},y_{3})\right\vert
^{2}dx_{1}dx_{2}dy_{1}dy_{2}dy_{3}\right) ^{\frac{1}{2}}\leqslant C.
\end{equation*}
\end{lemma}

\begin{remark}
Specializing to the case $n_{1},n_{2}=0,3,$or $6$, we will apply Lemma \ref%
{Lemma:KeyLemmaForErrorTerms} to prove Proposition \ref%
{Prop:ErrorEstimateConclusion}.
\end{remark}

In addition to the endpoint Strichartz estimates \cite{KeelAndTao} which are
necessary, we need the following estimate to prove Lemma \ref%
{Lemma:KeyLemmaForErrorTerms}.

\begin{claim}
\bigskip \label{Lemma:TraceLemmaForErrorTermsEstimates}%
\begin{equation*}
\left\Vert \int v_{0}^{2}(x-y)v_{0}^{2}(x-z)f(y,z)dydz\right\Vert _{L^{\frac{%
3}{2}}(\mathbb{R}^{3},dx)}\leqslant C\left\Vert f\right\Vert _{L^{\frac{3}{2}%
}(\mathbb{R}^{6},dydz)}
\end{equation*}
\end{claim}

\begin{proof}
\begin{eqnarray*}
&&\left\Vert \int v_{0}^{2}(x-y)v_{0}^{2}(x-z)f(y,z)dydz\right\Vert _{L^{%
\frac{3}{2}}(\mathbb{R}^{3},dx)} \\
&\leqslant &\left\Vert \int v_{0}^{2}(x-y)\left\Vert v_{0}^{2}\right\Vert
_{L^{3}}\left\Vert f(y,\cdot )\right\Vert _{L^{\frac{3}{2}}}dy\right\Vert
_{L^{\frac{3}{2}}} \\
&\leqslant &\left\Vert v_{0}^{2}\right\Vert _{L^{1}}\left\Vert
v_{0}^{2}\right\Vert _{L^{3}}\left\Vert f\right\Vert _{L^{\frac{3}{2}%
}}=C\left\Vert f\right\Vert _{L^{\frac{3}{2}}}.
\end{eqnarray*}
\end{proof}

We can prove Lemma \ref{Lemma:KeyLemmaForErrorTerms} now.

\begin{proof}
By Duhamel's principle, it suffices to prove%
\begin{eqnarray*}
&&\int dt\bigg(\int \left\vert e^{it(\pm \triangle _{x_{1}}\pm \triangle
_{x_{2}})}f(x_{1},x_{2})\right\vert ^{2}\left\vert e^{it(\pm \left( \Delta
_{y_{1}}+\Delta _{y_{2}}\right) \pm \Delta
_{y_{3}})}g(y_{1},y_{2},y_{3})\right\vert ^{2} \\
&&v^{2}(x_{1}-y_{1},x_{1}-y_{2})dx_{1}dx_{2}dy_{1}dy_{2}dy_{3}\bigg)^{\frac{1%
}{2}}\leqslant C\left\Vert f\right\Vert _{L^{2}}\left\Vert g\right\Vert
_{L^{2}}
\end{eqnarray*}%
because we have 
\begin{eqnarray*}
&&\Vert \left( i\frac{\partial }{\partial t}\pm \Delta _{x_{1}}\pm \Delta
_{x_{2}}\pm \left( \Delta _{y_{1}}+\Delta _{y_{2}}\right) \pm \Delta
_{y_{3}}\right) f(t,x_{1},x_{2})g(t,y_{1},y_{2},y_{3})\Vert _{L_{t}^{1}(%
\mathbb{R}^{+})L_{(x,y)}^{2}} \\
&\leqslant &C
\end{eqnarray*}%
with $f|_{t=0},g|_{t=0}\in L^{2}$ which also guarantees $f,g\in
L_{t}^{\infty }L_{x}^{2}$ by the energy estimate.

The proof is divided into two steps.

Step I: Write the partial Fourier transform to be%
\begin{equation*}
f_{\xi _{2}}^{\prime }(x_{1})=\int e^{ix_{1}\xi _{1}}\hat{f}(\xi _{1},\xi
_{2})d\xi _{1},
\end{equation*}%
then we have%
\begin{eqnarray*}
&&\int dx_{2}\left\vert e^{it(\pm \triangle _{x_{1}}\pm \triangle
_{x_{2}})}f(x_{1},x_{2})\right\vert ^{2} \\
&=&\int dx_{2}\int d\xi _{1}d\xi _{1}^{\prime }d\xi _{2}d\xi _{2}^{\prime
}e^{ix_{1}\left( \xi _{1}-\xi _{1}^{\prime }\right) }e^{it(\pm 1)(\left\vert
\xi _{1}\right\vert ^{2}-\left\vert \xi _{1}^{\prime }\right\vert
^{2})}e^{ix_{2}\left( \xi _{2}-\xi _{2}^{\prime }\right) }e^{it(\pm
1)(\left\vert \xi _{2}\right\vert ^{2}-\left\vert \xi _{2}^{\prime
}\right\vert ^{2})}\hat{f}(\xi _{1},\xi _{2})\overline{\hat{f}(\xi
_{1}^{\prime },\xi _{2}^{\prime })} \\
&=&\int d\xi _{1}d\xi _{1}^{\prime }d\xi _{2}d\xi _{2}^{\prime
}e^{ix_{1}\left( \xi _{1}-\xi _{1}^{\prime }\right) }e^{it(\pm 1)(\left\vert
\xi _{1}\right\vert ^{2}-\left\vert \xi _{1}^{\prime }\right\vert
^{2})}\delta (\xi _{2}-\xi _{2}^{\prime })e^{it(\pm 1)(\left\vert \xi
_{2}\right\vert ^{2}-\left\vert \xi _{2}^{\prime }\right\vert ^{2})}\hat{f}%
(\xi _{1},\xi _{2})\overline{\hat{f}(\xi _{1}^{\prime },\xi _{2}^{\prime })}
\\
&=&\int d\xi _{2}\int d\xi _{1}d\xi _{1}^{\prime }e^{ix_{1}\left( \xi
_{1}-\xi _{1}^{\prime }\right) }e^{it(\pm 1)(\left\vert \xi _{1}\right\vert
^{2}-\left\vert \xi _{1}^{\prime }\right\vert ^{2})}\hat{f}(\xi _{1},\xi
_{2})\overline{\hat{f}(\xi _{1}^{\prime },\xi _{2})} \\
&=&\int d\xi _{2}\left\vert e^{\pm it\triangle _{x_{1}}}f_{\xi _{2}}^{\prime
}(x_{1})\right\vert ^{2}.
\end{eqnarray*}

Step II: Let $\xi _{2},\eta _{3}$ be the phase variables corresponding to $%
x_{2},y_{3}.$ Utilizing H\"{o}lder and Claim \ref%
{Lemma:TraceLemmaForErrorTermsEstimates}, we get 
\begin{eqnarray*}
&&\int dt(\int \left\vert e^{it(\pm \triangle _{x_{1}}\pm \triangle
_{x_{2}})}f(x_{1},x_{2})\right\vert ^{2}\left\vert e^{it(\pm \left( \Delta
_{y_{1}}+\Delta _{y_{2}}\right) \pm \Delta
_{y_{3}})}g(y_{1},y_{2},y_{3})\right\vert ^{2} \\
&&v^{2}(x_{1}-y_{1},x_{1}-y_{2})dx_{1}dx_{2}dy_{1}dy_{2}dy_{3})^{\frac{1}{2}}
\\
&\leqslant &3\int dt\left( \int \left\vert e^{\pm it\triangle
_{x_{1}}}f_{\xi _{2}}^{\prime }(x_{1})\right\vert ^{2}\left\vert e^{\pm
it\left( \Delta _{y_{1}}+\Delta _{y_{2}}\right) }g_{\eta _{3}}^{\prime
}(y_{1},y_{2})\right\vert
^{2}v_{0}^{2}(x_{1}-y_{1})v_{0}^{2}(x_{1}-y_{2})dx_{1}dy_{1}dy_{2}d\xi
_{2}d\eta _{3}\right) ^{\frac{1}{2}} \\
&\leqslant &C\int \left( \int d\eta _{3}\left\Vert \int \left\vert e^{\pm
it\left( \Delta _{y_{1}}+\Delta _{y_{2}}\right) }g_{\eta _{3}}^{\prime
}(y_{1},y_{2})\right\vert
^{2}v_{0}^{2}(x_{1}-y_{1})v_{0}^{2}(x_{1}-y_{2})dy_{1}dy_{2}\right\Vert
_{L_{x_{1}}^{\frac{3}{2}}}\right) ^{\frac{1}{2}} \\
&&\left( \int d\xi _{2}\left\Vert \left\vert e^{\pm it\triangle
_{x_{1}}}f_{\xi _{2}}^{\prime }(x_{1})\right\vert ^{2}\right\Vert
_{L_{x_{1}}^{3}}\right) ^{\frac{1}{2}}dt \\
&\leqslant &C\int dt\left( \int d\xi _{2}\left\Vert e^{\pm it\triangle
_{x_{1}}}f_{\xi _{2}}^{\prime }(x_{1})\right\Vert
_{L_{x_{1}}^{6}}^{2}\right) ^{\frac{1}{2}}\left( \int d\eta _{3}\left\Vert
e^{\pm it\left( \Delta _{y_{1}}+\Delta _{y_{2}}\right) }g_{\eta
_{3}}^{\prime }(y_{1},y_{2})\right\Vert _{L_{(y_{1},y_{2})}^{3}}^{2}\right)
^{\frac{1}{2}}\text{ } \\
&\leqslant &C\left( \int dt\int d\xi _{2}\left\Vert e^{\pm it\triangle
_{x_{1}}}f_{\xi _{2}}^{\prime }(x_{1})\right\Vert
_{L_{x_{1}}^{6}}^{2}\right) ^{\frac{1}{2}}\left( \int dt\int d\eta
_{3}\left\Vert e^{\pm it\left( \Delta _{y_{1}}+\Delta _{y_{2}}\right)
}g_{\eta _{3}}^{\prime }(y_{1},y_{2})\right\Vert
_{L_{(y_{1},y_{2})}^{3}}^{2}\right) ^{\frac{1}{2}} \\
&\leqslant &C\left\Vert f\right\Vert _{L^{2}}\left\Vert g\right\Vert _{L^{2}}%
\text{ }\left( \text{endpoint Strichartz \cite{KeelAndTao}}\right)
\end{eqnarray*}%
The endpoint Strichartz estimates we used in the last line are the 3d $%
L_{t}^{2}L_{x}^{6}$ and the 6d $L_{t}^{2}L_{x}^{3}$ estimates.
\end{proof}

\subsection{Error term $e^{B}Ve^{-B}\Omega $, An Example}

Write 
\begin{eqnarray}
&&e^{B}Ve^{-B}  \notag \\
&=&\int dx_{0}dy_{0}dz_{0}v(x_{0}-y_{0},x_{0}-z_{0})  \notag \\
&&e^{B}a_{x_{0}}^{\ast }e^{-B}e^{B}a_{y_{0}}^{\ast
}e^{-B}e^{B}a_{z_{0}}^{\ast
}e^{-B}e^{B}a_{x_{0}}e^{-B}e^{B}a_{y_{0}}e^{-B}e^{B}a_{z_{0}}e^{-B}
\label{creation and annilation products in eVe}
\end{eqnarray}%
\begin{eqnarray}
&=&\int
dx_{0}dy_{0}dz_{0}dx_{1}dy_{1}dz_{1}dx_{2}dy_{2}dz_{2}v(x_{0}-y_{0},x_{0}-z_{0})
\notag \\
&&\left( u(x_{1},x_{0})a_{x_{1}}+\overline{\cosh (k)}%
(x_{1},x_{0})a_{x_{1}}^{\ast }\right) \left( u(y_{1},y_{0})a_{y_{1}}+%
\overline{\cosh (k)}(y_{1},y_{0})a_{y_{1}}^{\ast }\right)  \notag \\
&&\left( u(z_{1},z_{0})a_{z_{1}}+\overline{\cosh (k)}%
(z_{1},z_{0})a_{z_{1}}^{\ast }\right) \left( \cosh (k)(x_{2},x_{0})a_{x_{2}}+%
\overline{u}(x_{2},x_{0})a_{x_{2}}^{\ast }\right)  \notag \\
&&\left( \cosh (k)(y_{2},y_{0})a_{y_{2}}+\overline{u}%
(y_{2},y_{0})a_{y_{2}}^{\ast }\right) \left( \cosh (k)(z_{2},z_{0})a_{z_{2}}+%
\overline{u}(z_{2},z_{0})a_{z_{2}}^{\ast }\right)  \notag
\end{eqnarray}%
Because we are applying $e^{B}Ve^{-B}$ to $\Omega $, we neglect the terms in
product \ref{creation and annilation products in eVe} which have more
annihilation operators than creation operators. It is also unnecessary to
consider terms ending with $a_{z_{2}}$ or $a_{x_{2}}a_{y_{2}}a_{z_{2}}^{\ast
}$. These facts imply that $e^{B}Ve^{-B}\Omega $ has nonzero elements solely
in its 0th, 2nd, 4th and 6th Fock space slots. To exemplify the use of Lemma %
\ref{Lemma:KeyLemmaForErrorTerms}, we estimate two typical terms: the order
6 term%
\begin{eqnarray*}
&&\int dx_{0}dy_{0}dz_{0}dx_{1}dy_{1}dz_{1}dx_{2}dy_{2}dz_{2} \\
&&v(x_{0}-y_{0},x_{0}-z_{0})\overline{\cosh (k)}(x_{1},x_{0})\overline{\cosh
(k)}(y_{1},y_{0})\overline{\cosh (k)} \\
&&(z_{1},z_{0})\overline{u}(x_{2},x_{0})\overline{u}(y_{2},y_{0})\overline{u}%
(z_{2},z_{0})a_{x_{1}}^{\ast }a_{y_{1}}^{\ast }a_{z_{1}}^{\ast
}a_{x_{2}}^{\ast }a_{y_{2}}^{\ast }a_{z_{2}}^{\ast }
\end{eqnarray*}%
which contributes to the 6th Fock space slot of $e^{B}Ve^{-B}\Omega $ as%
\begin{eqnarray*}
&&\psi _{6}(x_{1},y_{1},z_{1},x_{2},y_{2},z_{2}) \\
&=&\int dx_{0}dy_{0}dz_{0}v(x_{0}-y_{0},x_{0}-z_{0}) \\
&&\overline{\cosh (k)}(x_{1},x_{0})\overline{\cosh (k)}(y_{1},y_{0})%
\overline{\cosh (k)}(z_{1},z_{0}) \\
&&\overline{u}(x_{2},x_{0})\overline{u}(y_{2},y_{0})\overline{u}%
(z_{2},z_{0}),
\end{eqnarray*}%
and an order 4 term 
\begin{eqnarray*}
&&\int
dx_{0}dy_{0}dz_{0}dx_{1}dy_{1}dz_{1}dx_{2}dy_{2}dz_{2}v(x_{0}-y_{0},x_{0}-z_{0})
\\
&&\overline{\cosh (k)}(x_{1},x_{0})\overline{\cosh (k)}(y_{1},y_{0})%
\overline{\cosh (k)}(z_{1},z_{0})\overline{u}(x_{2},x_{0})\cosh
(k)(y_{2},y_{0})\overline{u}(z_{2},z_{0}) \\
&&a_{x_{1}}^{\ast }a_{y_{1}}^{\ast }a_{z_{1}}^{\ast }a_{x_{2}}^{\ast
}a_{y_{2}}a_{z_{2}}^{\ast } \\
&=&\int
dx_{0}dy_{0}dz_{0}dx_{1}dy_{1}dz_{1}dx_{2}dy_{2}dz_{2}v(x_{0}-y_{0},x_{0}-z_{0})
\\
&&\overline{\cosh (k)}(x_{1},x_{0})\overline{\cosh (k)}(y_{1},y_{0})%
\overline{\cosh (k)}(z_{1},z_{0})\overline{u}(x_{2},x_{0})\cosh
(k)(y_{2},y_{0})\overline{u}(z_{2},z_{0}) \\
&&(\delta (y_{2}-z_{2})a_{x_{1}}^{\ast }a_{y_{1}}^{\ast }a_{z_{1}}^{\ast
}a_{x_{2}}^{\ast }+a_{x_{1}}^{\ast }a_{y_{1}}^{\ast }a_{z_{1}}^{\ast
}a_{x_{2}}^{\ast }a_{z_{2}}^{\ast }a_{y_{2}})
\end{eqnarray*}%
which contributes to the 4th Fock space slot of $e^{B}Ve^{-B}\Omega $ as%
\begin{eqnarray*}
\psi _{4}(x_{1},y_{1},z_{1},x_{2}) &=&\int
dx_{0}dy_{0}dz_{0}dy_{2}v(x_{0}-y_{0},x_{0}-z_{0}) \\
&&\overline{\cosh (k)}(x_{1},x_{0})\overline{\cosh (k)}(y_{1},y_{0})%
\overline{\cosh (k)}(z_{1},z_{0}) \\
&&\cosh (k)(y_{2},y_{0})\overline{u}(x_{2},x_{0})\overline{u}(y_{2},z_{0})
\end{eqnarray*}%
neglecting symmetrization and normalization.

\subsubsection{Estimate of $\protect\psi _{6}$, a triple product involving
one $u$}

Via the fact that%
\begin{equation*}
\cosh (k)(x,y)=\delta (x-y)+p(x,y)
\end{equation*}%
we write out the product in $\psi _{6}$ as%
\begin{equation*}
\psi _{6}=\psi _{6,\delta \delta \delta }+\psi _{6,p\delta \delta }+\psi
_{6,pp\delta }+\psi _{6,ppp}
\end{equation*}%
according to the factors of $\cosh $ carried in each term i.e.%
\begin{eqnarray*}
\psi _{6,\delta \delta \delta } &=&\int v(x_{0}-y_{0},x_{0}-z_{0})\delta
(x_{1}-x_{0})\delta (y_{1}-y_{0})\delta (z_{1}-z_{0})\overline{u}%
(x_{2},x_{0})\overline{u}(y_{2},y_{0})\overline{u}%
(z_{2},z_{0})dx_{0}dy_{0}dz_{0} \\
&=&v(x_{1}-y_{1},x_{1}-z_{1})\overline{u}(x_{2},x_{1})\overline{u}%
(y_{2},y_{1})\overline{u}(z_{2},z_{1})
\end{eqnarray*}%
and%
\begin{equation*}
\psi _{6,ppp}=\int v(x_{0}-y_{0},x_{0}-z_{0})\overline{p}(x_{1},x_{0})%
\overline{p}(y_{1},y_{0})\overline{p}(z_{1},z_{0})\overline{u}(x_{2},x_{0})%
\overline{u}(y_{2},y_{0})\overline{u}(z_{2},z_{0})dx_{0}dy_{0}dz_{0}
\end{equation*}%
etc. We proceed to estimate the worst term:%
\begin{eqnarray*}
&&\int dt\left( \int \left\vert \psi _{6,\delta \delta \delta }\right\vert
^{2}dx_{1}dy_{1}dz_{1}dx_{2}dy_{2}dz_{2}\right) ^{\frac{1}{2}} \\
&=&\int dt\left( \int \left\vert v(x_{1}-y_{1},x_{1}-z_{1})\overline{u}%
(x_{2},x_{1})\overline{u}(y_{2},y_{1})\overline{u}(z_{2},z_{1})\right\vert
^{2}dx_{1}dy_{1}dz_{1}dx_{2}dy_{2}dz_{2}\right) ^{\frac{1}{2}} \\
&\leqslant &C
\end{eqnarray*}%
where $\overline{u}(y_{2},y_{1})\overline{u}(z_{2},z_{1})$ takes the place
of $g$ in Lemma \ref{Lemma:KeyLemmaForErrorTerms}.

For terms in $\psi _{6}$ involving $p,$ we deal with them as the following:
By Cauchy-Schwarz on $dx_{0}dy_{0}dz_{0}$, we obtain%
\begin{eqnarray*}
&&\int \left( \int \left\vert \psi _{6,ppp}\right\vert
^{2}dx_{1}dy_{1}dz_{1}dx_{2}dy_{2}dz_{2}\right) ^{\frac{1}{2}}dt \\
&\leqslant &\sup_{t}(\int \left\vert
p(x_{1},x_{0})p(y_{1},y_{0})p(z_{1},z_{0})\right\vert
^{2}dx_{0}dy_{0}dz_{0}dx_{1}dy_{1}dz_{1})^{\frac{1}{2}} \\
&&\int \left( \int \left\vert
v(x_{0}-y_{0},x_{0}-z_{0})u(x_{2},x_{0})u(y_{2},y_{0})u(z_{2},z_{0})\right%
\vert ^{2}dx_{0}dy_{0}dz_{0}dx_{2}dy_{2}dz_{2}\right) ^{\frac{1}{2}}dt
\end{eqnarray*}%
where the first integral is majorized by the energy estimate of $p,$ the
second integral is the same as the one appearing in $\psi _{6,\delta \delta
\delta }$ and can be taken care of by Lemma \ref{Lemma:KeyLemmaForErrorTerms}%
.

\begin{remark}
In the estimate regarding $\psi _{6,ppp},$ we can do Cauchy-Schwarz in
another way:%
\begin{eqnarray*}
&&\int \left( \int \left\vert \psi _{6,ppp}\right\vert
^{2}dx_{1}dy_{1}dz_{1}dx_{2}dy_{2}dz_{2}\right) ^{\frac{1}{2}}dt \\
&\leqslant &\sup_{t}(\int \left\vert
p(x_{1},x_{0})u(y_{2},y_{0})u(z_{2},z_{0})\right\vert
^{2}dx_{0}dy_{0}dz_{0}dx_{1}dy_{2}dz_{2})^{\frac{1}{2}} \\
&&\int \left( \int \left\vert
v(x_{0}-y_{0},x_{0}-z_{0})u(x_{2},x_{0})p(y_{1},y_{0})p(z_{1},z_{0})\right%
\vert ^{2}dx_{0}dy_{0}dz_{0}dy_{1}dz_{1}dx_{2}\right) ^{\frac{1}{2}}dt
\end{eqnarray*}%
which also works by Lemma \ref{Lemma:KeyLemmaForErrorTerms}. Because $%
\left\Vert u\right\Vert \geqslant \left\Vert p\right\Vert $
\end{remark}

\subsubsection{Estimate of $\protect\psi _{4}$, a double product involving
one $u$}

\begin{eqnarray*}
\psi _{4}(x_{1},y_{1},z_{1},x_{2}) &=&\int
dx_{0}dy_{0}dz_{0}dy_{2}v(x_{0}-y_{0},x_{0}-z_{0}) \\
&&\overline{\cosh (k)}(x_{1},x_{0})\overline{\cosh (k)}(y_{1},y_{0})%
\overline{\cosh (k)}(z_{1},z_{0}) \\
&&\cosh (k)(y_{2},y_{0})\overline{u}(x_{2},x_{0})\overline{u}(y_{2},z_{0}) \\
&=&\psi _{4,\delta \delta \delta \delta }+...+\psi _{4,pppp}
\end{eqnarray*}%
where the worst term is%
\begin{eqnarray}
\psi _{4,\delta \delta \delta \delta } &=&\int
dx_{0}dy_{0}dz_{0}dy_{2}v(x_{0}-y_{0},x_{0}-z_{0})
\label{formula:Worst3-bodyErrorTerm} \\
&&\delta (x_{1}-x_{0})\delta (y_{1}-y_{0})\delta (z_{1}-z_{0})\delta
(y_{2}-y_{0})\overline{u}(x_{2},x_{0})\overline{u}(y_{2},z_{0})  \notag \\
&=&v(x_{1}-y_{1},x_{1}-z_{1})\bar{u}(x_{2},x_{1})\overline{u}(y_{1},z_{1}). 
\notag
\end{eqnarray}%
Letting $\overline{u}(y_{1},z_{1})$ be $g$ in Lemma \ref%
{Lemma:KeyLemmaForErrorTerms}, we derive the desired estimate%
\begin{equation*}
\int dt\left( \int dx_{1}dy_{1}dz_{1}dx_{2}\left\vert \psi _{4,\delta \delta
\delta \delta }\right\vert ^{2}\right) ^{\frac{1}{2}}\leqslant C.
\end{equation*}

\subsection{Remark for all other error terms}

At a glance, we can handle all terms using Lemma \ref%
{Lemma:KeyLemmaForErrorTerms}, except 
\begin{equation*}
\int v(x-y,x-z)\phi (x)\phi (y)\phi (z)a_{x}^{\ast }a_{y}^{\ast }a_{z}^{\ast
}dxdydz
\end{equation*}%
in $[A,[A,[A,V]]]$, since all other terms end with $a$ instead of $a^{\ast }$%
. This observation allows the application of Lemma \ref%
{Lemma:KeyLemmaForErrorTerms}. But Lemma \ref{Lemma:KeyLemmaForErrorTerms}
also applies to 
\begin{equation*}
e^{B}\left( \int v(x-y,x-z)\phi (x)\phi (y)\phi (z)a_{x}^{\ast }a_{y}^{\ast
}a_{z}^{\ast }dxdydz\right) e^{-B}\Omega .
\end{equation*}%
because we can let $\phi (x_{1})$ be $f(x_{1})$, $\phi (y_{1})\phi (y_{2})$
be $g(y_{1},y_{2})$

Therefore we have established Proposition \ref{Prop:ErrorEstimateConclusion}
and thus Theorem \ref{THM:MainTheoremErrorEstimate}.

\section{The Long Time Behavior of The Hartree Equation / Proof of Theorem 
\protect\ref{THM:TheLongTimeBehaviorOfHartree}}

\label{section:AppendixHartree}

In this section, we discuss the Hartree equation \ref%
{Equation:DefocusingHartreeEquation}%
\begin{equation*}
i\frac{\partial }{\partial t}\phi +\triangle \phi -\frac{1}{2}\phi \int
v(x-y,x-z)\left\vert \phi (y)\right\vert ^{2}\left\vert \phi (z)\right\vert
^{2}dydz=0
\end{equation*}%
where%
\begin{equation*}
v(x-y,x-z)=v_{0}(x-y)v_{0}(x-z)+v_{0}(x-y)v_{0}(y-z)+v_{0}(x-z)v_{0}(y-z).
\end{equation*}%
We assume the nonnegative regular potential $v_{0}$ decays fast enough away
from the origin and has the property that 
\begin{equation*}
v_{0}(x)=v_{0}(-x).
\end{equation*}

Throughout this section, we write 
\begin{eqnarray*}
A &=&\int v_{0}(x-y)v_{0}(x-z)\left\vert \phi (y)\right\vert ^{2}\left\vert
\phi (z)\right\vert ^{2}dydz \\
B &=&\int v_{0}(x-y)v_{0}(y-z)\left\vert \phi (y)\right\vert ^{2}\left\vert
\phi (z)\right\vert ^{2}dydz \\
C &=&\int v_{0}(x-z)v_{0}(y-z)\left\vert \phi (y)\right\vert ^{2}\left\vert
\phi (z)\right\vert ^{2}dydz,
\end{eqnarray*}%
for convenience i.e. equation \ref{Equation:DefocusingHartreeEquation}
becomes 
\begin{equation}
i\frac{\partial }{\partial t}\phi +\triangle \phi -\frac{1}{2}\left( A\phi
-B\phi -C\phi \right) =0.  \label{Equation:HartreeEquationInAppendix}
\end{equation}

So (ii) becomes%
\begin{equation*}
E(t)|_{t=0}=\left( \frac{1}{2}\int \left\vert \nabla \phi \right\vert ^{2}+%
\frac{1}{6}\int (A+B+C)\left\vert \phi \right\vert ^{2}\right)
|_{t=0}<\infty \text{ }.
\end{equation*}%
and (i)-(iii) implies%
\begin{equation*}
E_{c}(t)=\int t^{2}\left( \left\vert \nabla (e^{i\frac{\left\vert
x\right\vert ^{2}}{4t}}\phi )\right\vert ^{2}+\frac{1}{6}(A+B+C)\left\vert
\phi \right\vert ^{2}\right) <\infty
\end{equation*}

To prove Theorem \ref{THM:TheLongTimeBehaviorOfHartree}, we are going to
argue that%
\begin{equation*}
\dot{E}_{c}(t)\leqslant 0\text{ for }t\geqslant 1
\end{equation*}%
which leads to%
\begin{equation*}
\left\Vert \phi \right\Vert _{L^{6}}\leqslant C\left\Vert \nabla (e^{i\frac{%
\left\vert x\right\vert ^{2}}{4t}}\phi )\right\Vert _{L^{2}}\leqslant
C\left( \frac{E_{c}(t)}{t^{2}}\right) ^{\frac{1}{2}}\leqslant \frac{C}{t}%
\text{ for }t\geqslant 1.
\end{equation*}%
Here are the details of Theorem \ref{THM:TheLongTimeBehaviorOfHartree}.

\subsection{Conservation of Mass, Momentum, and Energy}

First, it is not difficult to see the conservation law of the $L^{2}$ mass 
\begin{equation}
\partial _{t}\rho -\nabla _{j}p^{j}=0  \label{Conservation:L^2}
\end{equation}%
where%
\begin{equation*}
\rho :=\frac{1}{2}\left\vert \phi \right\vert ^{2}
\end{equation*}%
and%
\begin{equation*}
p_{j}:=\frac{1}{2i}(\phi \nabla _{j}\bar{\phi}-\bar{\phi}\nabla _{j}\phi )
\end{equation*}%
because equation \ref{Equation:HartreeEquationInAppendix} is of the form%
\begin{equation*}
i\frac{\partial }{\partial t}\phi +\triangle \phi =F(\left\vert \phi
\right\vert ^{2})\phi .
\end{equation*}%
Times $-\bar{\phi}$ to equation \ref{Equation:HartreeEquationInAppendix}, we
acquire%
\begin{equation*}
-p_{0}+\frac{1}{2}\sigma -\triangle \rho +(A+B+C)\rho =0
\end{equation*}%
where%
\begin{eqnarray*}
p_{0} &:&=\frac{1}{2i}(\phi \bar{\phi}_{t}-\bar{\phi}\phi _{t}) \\
\sigma &:&=tr(\sigma _{jk})=tr(\nabla _{j}\bar{\phi}\nabla _{k}\phi +\nabla
_{k}\bar{\phi}\nabla _{j}\phi )
\end{eqnarray*}%
Moreover, letting%
\begin{eqnarray*}
\lambda &:&=(-p_{0}+\frac{1}{2}\sigma +\frac{1}{3}(A+B+C)\rho )=\triangle
\rho -\frac{2}{3}(A+B+C)\rho \\
e &:&=\frac{1}{2}\sigma +\frac{1}{3}(A+B+C)\rho
\end{eqnarray*}%
produces the conservation law of energy%
\begin{equation}
\partial _{t}e-\nabla _{j}\sigma _{0}^{j}+l_{0}=0
\label{Conservation:Energy}
\end{equation}%
where%
\begin{eqnarray*}
\sigma _{0}^{j} &:&=\phi _{t}\nabla _{j}\bar{\phi}+\bar{\phi}_{t}\nabla
_{j}\phi \\
l_{0} &=&\frac{2}{3}(A+B+C)\rho _{t}-\frac{1}{3}(A+B+C)_{t}\rho .
\end{eqnarray*}%
A direct computation shows that 
\begin{eqnarray*}
\int A_{t}\rho &=&\int v_{0}(x-y)v_{0}(x-z)(\left\vert \phi (y)\right\vert
^{2}\left\vert \phi (z)\right\vert ^{2})_{t}(\frac{\left\vert \phi
(x)\right\vert ^{2}}{2})dxdydz \\
&=&\int v_{0}(x-y)v_{0}(y-z)\left\vert \phi (y)\right\vert ^{2}\left\vert
\phi (z)\right\vert ^{2}(\frac{\left\vert \phi (x)\right\vert ^{2}}{2}%
)_{t}dxdydz \\
&&+\int v_{0}(x-z)v_{0}(y-z)\left\vert \phi (y)\right\vert ^{2}\left\vert
\phi (z)\right\vert ^{2}(\frac{\left\vert \phi (x)\right\vert ^{2}}{2}%
)_{t}dxdydz
\end{eqnarray*}%
and%
\begin{eqnarray*}
\int B_{t}\rho &=&\int v_{0}(x-y)v_{0}(x-z)\left\vert \phi (y)\right\vert
^{2}\left\vert \phi (z)\right\vert ^{2}(\frac{\left\vert \phi (x)\right\vert
^{2}}{2})_{t}dxdydz \\
&&+\int v_{0}(x-z)v_{0}(y-z)\left\vert \phi (y)\right\vert ^{2}\left\vert
\phi (z)\right\vert ^{2}(\frac{\left\vert \phi (x)\right\vert ^{2}}{2}%
)_{t}dxdydz
\end{eqnarray*}%
\begin{eqnarray*}
\int C_{t}\rho &=&\int v_{0}(x-y)v_{0}(y-z)\left\vert \phi (y)\right\vert
^{2}\left\vert \phi (z)\right\vert ^{2}(\frac{\left\vert \phi (x)\right\vert
^{2}}{2})_{t}dxdydz \\
&&+\int v_{0}(x-y)v_{0}(x-z)\left\vert \phi (y)\right\vert ^{2}\left\vert
\phi (z)\right\vert ^{2}(\frac{\left\vert \phi (x)\right\vert ^{2}}{2}%
)_{t}dxdydz
\end{eqnarray*}%
i.e.%
\begin{equation*}
\int (A+B+C)_{t}\rho =2\int (A+B+C)\rho _{t}
\end{equation*}%
which implies $\int l_{0}=0$ i.e. the conservation of energy 
\begin{equation*}
E(t)=\frac{1}{2}\int \left\vert \nabla \phi \right\vert ^{2}+\frac{1}{6}\int
(A+B+C)\left\vert \phi \right\vert ^{2}
\end{equation*}

Similarly, we derive the conservation law of momentum:%
\begin{equation}
\partial _{t}p_{j}-\nabla _{k}\left\{ \sigma _{j}^{j}-\delta _{j}^{k}\lambda
\right\} +l_{j}=0  \label{Conservation:Momentum}
\end{equation}%
where%
\begin{equation*}
l_{j}:=\frac{2}{3}(A+B+C)\rho _{j}-\frac{1}{3}(A+B+C)_{j}\rho .
\end{equation*}

\subsection{Conformal Identity}

At this point, if we multiply conservation law \ref{Conservation:L^2} by $%
\frac{\left\vert x\right\vert ^{2}}{2},$ \ref{Conservation:Momentum} by $%
tx^{j}$ and \ref{Conservation:Energy} by $t^{2}$ and add the resulting
identities, we obtain the conformal identity: 
\begin{equation*}
\partial _{t}e_{c}-\nabla _{j}\tau ^{j}+r=0
\end{equation*}%
where%
\begin{eqnarray*}
e_{c} &:&=(\frac{\left\vert x\right\vert ^{2}}{2})\rho
+tx^{j}p_{j}+t^{2}e=t^{2}\left( \left\vert \nabla (e^{i\frac{\left\vert
x\right\vert ^{2}}{4t}}\phi )\right\vert ^{2}+\frac{1}{3}(A+B+C)\rho \right)
\\
\tau ^{j} &:&=(\frac{\left\vert x\right\vert ^{2}}{2})p^{j}+tx^{k}\sigma
_{k}^{j}+tx^{j}\left( -\triangle \rho +\frac{2}{3}(A+B+C)\rho \right)
+t^{2}\sigma _{0}^{j} \\
r &:&=t^{2}l_{0}+tx^{j}l_{j}-nt\triangle \rho +t(n-1)\frac{2}{3}(A+B+C)\rho .
\end{eqnarray*}

This suggests%
\begin{equation}
\dot{E}_{c}+R_{c}=0  \label{ConformalIdentity}
\end{equation}%
where%
\begin{equation*}
R_{c}:=t\int \left( (n-1)\frac{2}{3}(A+B+C)\rho +x^{j}l_{j}\right) dx.
\end{equation*}

To determine $\dot{E}_{c}$, we calculate%
\begin{eqnarray*}
&&\int x^{j}l_{j}dx \\
&=&\frac{8}{3}\int x^{j}v(x-y,x-z)(\rho _{1})_{j}\rho _{2}\rho _{3}-\frac{8}{%
3}\int x^{j}v(x-y,x-z)\rho _{1}(\rho _{2})_{j}\rho _{3} \\
&=&\frac{8}{3}\int \rho _{3}v(x-y,x-z)x^{j}[(\rho _{1})_{j}\rho _{2}-\rho
_{1}(\rho _{2})_{j}] \\
&=&\frac{16}{3}\int \rho _{3}v(x-y,x-z)x\cdot \nabla _{1-2}(\rho _{1}\rho
_{2}) \\
&=&\frac{8}{3}\int \rho _{3}v(x-y,x-z)(x+y)\cdot \nabla _{1-2}(\rho _{1}\rho
_{2})+\frac{8}{3}\int \rho _{3}v(x-y,x-z)(x-y)\cdot \nabla _{1-2}(\rho
_{1}\rho _{2}) \\
&=&-\frac{8}{3}\int \rho _{1}\rho _{2}\rho _{3}\nabla _{1-2}v(x-y,x-z)\cdot
(x+y)-\frac{8}{3}\int \rho _{1}\rho _{2}\rho _{3}\nabla
_{1-2}v(x-y,x-z)\cdot (x-y) \\
&&-\frac{8}{3}n\int \rho _{1}\rho _{2}\rho _{3}v(x-y,x-z)
\end{eqnarray*}%
where $\nabla _{1-2}=\nabla _{x-y}=\frac{1}{2}(\nabla _{x}-\nabla _{y}).$

Insert formula \ref{Formula:DefinitonOfv-epsilon}%
\begin{equation*}
v(x-y,x-z)=v_{0}(x-y)v_{0}(x-z)+v_{0}(x-y)v_{0}(y-z)+v_{0}(x-z)v_{0}(y-z)
\end{equation*}%
to the above computation, it is%
\begin{eqnarray*}
&&\int x^{j}l_{j}dx \\
&=&-\frac{8}{3}\int \rho _{1}\rho _{2}\rho _{3}v_{0}(x-z)\left( \nabla
_{1-2}v_{0}(x-y)\right) \cdot (x+y)-\frac{4}{3}\int \rho _{1}\rho _{2}\rho
_{3}v_{0}(x-y)\left( \nabla _{x}v_{0}(x-z)\right) \cdot (x+y) \\
&&-\frac{8}{3}\int \rho _{1}\rho _{2}\rho _{3}v_{0}(y-z)\left( \nabla
_{1-2}v_{0}(x-y)\right) \cdot (x+y)+\frac{4}{3}\int \rho _{1}\rho _{2}\rho
_{3}v_{0}(x-y)\left( \nabla _{y}v_{0}(y-z)\right) \cdot (x+y) \\
&&-\frac{4}{3}\int \rho _{1}\rho _{2}\rho _{3}v_{0}(y-z)\left( \nabla
_{x}v_{0}(x-z)\right) \cdot (x+y)+\frac{4}{3}\int \rho _{1}\rho _{2}\rho
_{3}v_{0}(x-z)\left( \nabla _{y}v_{0}(y-z)\right) \cdot (x+y) \\
&&-\frac{2}{3}n\int (A+B+C)\rho \\
&&-\frac{8}{3}\int \rho _{1}\rho _{2}\rho _{3}v_{0}(x-z)\left( \nabla
_{1-2}v_{0}(x-y)\right) \cdot (x-y)-\frac{4}{3}\int \rho _{1}\rho _{2}\rho
_{3}v_{0}(x-y)\left( \nabla _{x}v_{0}(x-z)\right) \cdot (x-y) \\
&&-\frac{8}{3}\int \rho _{1}\rho _{2}\rho _{3}v_{0}(y-z)\left( \nabla
_{1-2}v_{0}(x-y)\right) \cdot (x-y)+\frac{4}{3}\int \rho _{1}\rho _{2}\rho
_{3}v_{0}(x-y)\left( \nabla _{y}v_{0}(y-z)\right) \cdot (x-y) \\
&&-\frac{4}{3}\int \rho _{1}\rho _{2}\rho _{3}v_{0}(y-z)\left( \nabla
_{x}v_{0}(x-z)\right) \cdot (x-y)+\frac{4}{3}\int \rho _{1}\rho _{2}\rho
_{3}v_{0}(x-z)\left( \nabla _{y}v_{0}(y-z)\right) \cdot (x-y).
\end{eqnarray*}

Notice that, in the above calculation. 
\begin{eqnarray*}
1st+3rd &=&0 \\
\left( 4th+11th\right) +\left( 6th+13th\right) &=&0 \\
\left( 2nd+9th\right) +\left( 5th+12th\right) &=&-\frac{8}{3}\int \rho
_{1}\rho _{2}\rho _{3}v_{0}(y-z)\left( \nabla _{x}v_{0}(x-z)\right) \cdot
(x-z).
\end{eqnarray*}%
So $\int x^{j}l_{j}dx$ simplifies to%
\begin{eqnarray*}
\int x^{j}l_{j}dx &=&-\frac{2}{3}n\int (A+B+C)\rho \\
&&-\frac{8}{3}\int \rho _{1}\rho _{2}\rho _{3}v_{0}(x-z)\left( \left( \nabla
v_{0}\right) (x-y)\right) \cdot (x-y) \\
&&-\frac{8}{3}\int \rho _{1}\rho _{2}\rho _{3}v_{0}(y-z)\left( \left( \nabla
v_{0}\right) (x-y)\right) \cdot (x-y) \\
&&-\frac{8}{3}\int \rho _{1}\rho _{2}\rho _{3}v_{0}(y-z)\left( \left( \nabla
v_{0}\right) (x-z)\right) \cdot (x-z).
\end{eqnarray*}
Hence%
\begin{eqnarray*}
R_{c} &=&t\int \left( (n-1)\frac{2}{3}(A+B+C)\rho +x^{j}l_{j}\right) dx \\
&=&-\frac{8}{3}t\int \rho _{1}\rho _{2}\rho _{3}v_{0}(x-z)\left\{
v_{0}(x-y)+\left( \left( \nabla v_{0}\right) (x-y)\right) \cdot (x-y)\right\}
\\
&&-\frac{8}{3}t\int \rho _{1}\rho _{2}\rho _{3}v_{0}(y-z)\left\{
v_{0}(x-y)+\left( \left( \nabla v_{0}\right) (x-y)\right) \cdot (x-y)\right\}
\\
&&-\frac{8}{3}t\int \rho _{1}\rho _{2}\rho _{3}v_{0}(y-z)\left\{
v_{0}(x-z)+\left( \left( \nabla v_{0}\right) (x-z)\right) \cdot
(x-z)\right\} .
\end{eqnarray*}%
When $v_{0}$ decays fast enough, we have 
\begin{equation*}
R_{c}\geqslant 0,
\end{equation*}%
or in other words 
\begin{equation*}
\dot{E}_{c}\leqslant 0\text{ for }t\geqslant 1,
\end{equation*}%
which implies $E_{c}(t)$ does not increase as claimed.

\section{Acknowledgements}

The author's thanks go to Professors Matei Machedon, Manoussos Grillakis and
Dionisios Margetis for discussions related to this work, and to the
anonymous referees for their many insightful comments and helpful
suggestions.

\end{document}